\documentclass{easychair}
\usepackage[utf8]{inputenc}
\usepackage{amsmath,amssymb,amsthm}
\usepackage{bm,cancel,xcolor,tikz,bussproofs,multicol,hyperref,multirow,mathtools,pgfplots}

\pgfplotsset{compat=1.18}
\usepackage{todonotes}
\usepackage{algorithm,longtable}
\usepackage{vampire_induction}
\usetikzlibrary{trees,decorations.pathmorphing,decorations.pathreplacing,shapes,shapes.geometric,arrows.meta,calc,patterns,patterns.meta,plotmarks,spy}
\usepackage{mdframed,wrapfig}
\usepackage[shortlabels]{enumitem}
\usepackage{thm-restate}
\usepackage{dashbox}
\usepackage{wasysym}

\tikzset{%
downa/.style={thick,-{Stealth}},
upa/.style={line width=0.5pt,double distance=0.8pt,double,-Stealth},
bluearrow/.style={thick,-{Stealth},blue},
redarrow/.style={thick,-{Stealth},dashed,red},
blackarrow/.style={thick,-{Stealth},dotted,very thick},
}
\tikzdeclarepattern{
  name=mylines,
  parameters={
      \pgfkeysvalueof{/pgf/pattern keys/size},
      \pgfkeysvalueof{/pgf/pattern keys/angle},
      \pgfkeysvalueof{/pgf/pattern keys/line width},
  },
  bounding box={
    (0,-0.5*\pgfkeysvalueof{/pgf/pattern keys/line width}) and
    (\pgfkeysvalueof{/pgf/pattern keys/size},
0.5*\pgfkeysvalueof{/pgf/pattern keys/line width})},
  tile size={(\pgfkeysvalueof{/pgf/pattern keys/size},
\pgfkeysvalueof{/pgf/pattern keys/size})},
  tile transformation={rotate=\pgfkeysvalueof{/pgf/pattern keys/angle}},
  defaults={
    size/.initial=5pt,
    angle/.initial=45,
    line width/.initial=.4pt,
  },
  code={
      \draw [line width=\pgfkeysvalueof{/pgf/pattern keys/line width}]
        (0,0) -- (\pgfkeysvalueof{/pgf/pattern keys/size},0);
  },
}

\theoremstyle{definition}
\newtheorem{definition}{Definition}
\theoremstyle{plain}
\newtheorem{theorem}{Theorem}
\newtheorem{example}{Example}
\newtheorem{lemma}{Lemma}
\newtheorem{remark}{Remark}

\AfterEndEnvironment{definition}{\noindent\ignorespaces}
\AfterEndEnvironment{theorem}{\noindent\ignorespaces}
\AfterEndEnvironment{example}{\noindent\ignorespaces}
\AfterEndEnvironment{lemma}{\noindent\ignorespaces}
\AfterEndEnvironment{restatable}{\noindent\ignorespaces}
\AfterEndEnvironment{remark}{\noindent\ignorespaces}

\newcommand{\Ind}{\mtt{Ind}}

\newcommand{\mgu}{\mtt{mgu}}
\def\orcidID#1{\href{http://orcid.org/#1}{\raisebox{-1.25pt}{\includegraphics{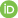}}}}

\newcommand{\bluearrow}[1]{\tikz{\draw[bluearrow](0,0) to (#1,0);}}
\newcommand{\redarrow}[1]{\tikz{\draw[redarrow](0,0) to (#1,0);}}
\newcommand{\blackarrow}[1]{\tikz{\draw[blackarrow](0,0) to (#1,0);}}

\DeclareUnicodeCharacter{25EA}{\halfsquare}
\DeclareRobustCommand\halfsquare{%
  \unskip\nobreak\thinspace\textemdash\allowbreak\thinspace\ignorespaces}

\makeatletter
\DeclareFontEncoding{LS2}{}{\noaccents@}
\DeclareFontSubstitution{LS2}{stix}{m}{n}
\DeclareSymbolFont{arrows3}{LS2}{stixtt}{m}{n}
\DeclareMathSymbol{\squarelrblack}{\mathord}{arrows3}{"89}
\makeatother

\newcommand\drw[1]{\ensuremath{{}^{\square} #1}}
\newcommand\erw[1]{\ensuremath{{}^{\squarelrblack} #1}}
\newcommand\urw[1]{\ensuremath{{}^{\blacksquare} #1}}

\newcommand{\Sup}{\ensuremath{\mathbf{Sup}}\xspace}
\newcommand{\ReC}{\ensuremath{\mathbf{ReC}}\xspace}
\newcommand{\CReC}{\ensuremath{\mathbf{CReC}}\xspace}
\newcommand{\Rw}{\mtt{Rw\xspace}}
\newcommand{\CRw}{\mtt{CRw\xspace}}
\newcommand{\Cleft}{\mtt{Chain_1\xspace}}
\newcommand{\Cright}{\mtt{Chain_2\xspace}}

\usepackage[only,llbracket,rrbracket]{stmaryrd}
\newcommand{\lenif}[1]{\llbracket#1\rrbracket}

\newcommand{\Axioms}{\mtt{Ax}}
\newcommand{\Conj}{\mtt{Conj}}
\newcommand{\Lem}{\mtt{Lem}}

\newcommand{\CommentedOut}[1]{}

\begin{document}

\title{Rewriting and Inductive Reasoning}
\author{M\'arton Hajdu\orcidID{0000-0002-8273-2613} \and
Laura Kovács\orcidID{0000-0002-8299-2714}
\and 
Michael Rawson\orcidID{0000-0001-7834-1567}
}
\institute{TU Wien}
\date{}
\titlerunning{Rewriting and Inductive Reasoning}
\authorrunning{M. Hajdu, L. Kovács, M. Rawson}

\maketitle
\vspace{-1.5em}
\begin{abstract} 
%
%
Rewriting techniques based on reduction orderings generate ``just enough'' consequences to retain first-order completeness.
This is ideal for superposition-based first-order theorem proving, but for at least one approach to inductive reasoning we show that we are missing crucial consequences.
We therefore extend the superposition calculus with rewriting-based techniques to generate sufficient consequences for automating induction in saturation. When applying our work within the unit-equational fragment, our experiments with the theorem prover Vampire show significant improvements for inductive reasoning.


\end{abstract}

\section{Introduction}\label{sec:introduction}

Automating \emph{proof by induction} is a particularly hard task with a long history~\cite{Aubin, ACL2, Rippling93, RRL, WaltherComputing}.
A recent promising line of research in this direction comes with the  integration of induction into saturation-based first-order theorem proving~\cite{CADE19,GettingSaturated22}, by extending the logical calculus with induction inference rules. Such induction rules are of the form
\begin{prooftree}
\AxiomC{$\neg L[t]\lor C$}
\UnaryInfC{$\cnf(\neg F\lor C)$}
\end{prooftree}
\noindent where $\neg L[t]\lor C$ is a clause, $\neg L[t]$ a ground literal, $t$ a term of some inductive data type, and $F\to \forall x.L[x]$ is an instance of some valid second-order induction \emph{schema}, for example structural induction\footnote{$\cnf(\neg F\lor C)$ denotes the clausified formula $\neg F\lor C$}. Note that 
the schema instance $F\to \forall x.L[x]$ is applied to resolve $\neg L[t]$, but the schema instance is never added to the saturation search space, making sure that the conclusion of the inference is derived via inductive reasoning. As such, the application of the induction rule is  \emph{triggered} by the presence of the literal $\neg L[t]$ in the search space. Hence, if there is no such literal, no induction is applied. 
Applying induction only when triggered means that only premises of schema instances directly related to a clause selected during proof search are added to the search space.
This is a strong heuristic method for automating induction, particularly in the presence of full first-order logic and theories~\cite{CICM20,CADE21,FMCAD21}.

\paragraph{Challenge.}
Unfortunately, applying induction only when triggered leads to tension between two competing factors:
\begin{enumerate}[label=(\roman*),leftmargin=1.7em]
\item
It may be that a valid goal is not provable without deriving a certain consequence $C$, that in   turn triggers an induction rule with premise $C$ with a specific schema. 
\item
Efficient first-order calculi, such as superposition~\cite{EquationalReasoningInSaturation}, go to great lengths to derive and retain only those consequences absolutely required for completeness. Therefore, induction rules that could be triggered by a (missing) consequence $C$ might not be applied.
\end{enumerate}
\noindent
We show in our motivating example  that this tension causes a lost proof in practice (Section~\ref{motivating-example}), as superposition may avoid generating consequences that would be needed to be used in inductive proofs. 

\paragraph{Our contributions.} For automating (triggered) induction in saturation, we thus need to do something counter-intuitive for those accustomed to first-order superposition reasoning. We propose deriving slightly \emph{more} consequences than usual, in order to trigger induction rules with suitable schemas. Naturally, this must still be as few extra consequences as possible in order to retain a high level of first-order efficiency.
Concretely, motivated by applications of program verification~\cite{DBLP:conf/fmcad/GeorgiouGBRKR22} in this paper we focus on the unit-equational fragment of first-order logic with induction and bring the following contributions. 
\begin{enumerate}[label=(\arabic*),leftmargin=1.7em]
\item
We introduce a modified superposition calculus (Section~\ref{sec:rewriting}) with slightly relaxed constraints, allowing us to derive consequences that cannot be derived by  standard superposition.
\item
We impose new restrictions on our calculus (Section~\ref{sec:rewriting}) to further improve efficiency while retaining our newfound ability to generate consequences.
\item
We improve redundancy elimination for saturation with induction, providing sufficient conditions for skipping redundant induction steps and rewrites (Section~\ref{sec:induction}).
\item
We implement our approach in the \textsc{Vampire} first-order prover~\cite{CAV13} and apply it to proving inductive properties (Section~\ref{sec:evaluation}).
Results show that our work improves the state of the art in automating proofs by induction.
\end{enumerate}
\begin{figure*}[t]
\setlength{\abovecaptionskip}{0pt}
\setlength{\belowcaptionskip}{0pt}
\small
\hrule\smallskip
{\bf Signature}
\smallskip\hrule\vspace{-1em}
\begin{minipage}{.35\linewidth}
\begin{align*}
\zero &: \nat\\
\suc &: \nat\to\nat\\
\nil &: \lst
\end{align*}
\end{minipage}\begin{minipage}{.35\linewidth}
\begin{align*}
\cons &: \nat \to \lst \to \lst\\
(+) &: \nat\to\nat\to\nat\\
(\appif) &: \lst\to\lst\to\lst
\end{align*}
\end{minipage}\begin{minipage}{.3\linewidth}
\begin{align*}
(\lenif{.}) &: \lst\to\nat\\
f &: \lst \to\lst\\
g &: \lst \to\lst
\end{align*}
\end{minipage}
\smallskip\hrule\smallskip
{\bf Axioms} --- \Axioms
\smallskip\hrule\vspace{-1em}
\begin{minipage}{.5\linewidth}
\begin{align}
    &\zero\not\simeq\suc(x)\label{eq:nat1}\tag{$\nat.1$}\\
    &\suc(x)\not\simeq\suc(y)\lor x\simeq y\label{eq:nat2}\tag{$\nat.2$}\\[.7em]
    &\nil\not\simeq\cons(x,y)\label{eq:list1}\tag{$\lst.1$}\\
    &\cons(x,y)\not\simeq\cons(z,u)\lor x\simeq z\label{eq:list2}\tag{$\lst.2$}\\
    &\cons(x,y)\not\simeq\cons(z,u)\lor y\simeq u\label{eq:list3}\tag{$\lst.3$}\\[.7em]
    &\zero + x \simeq x&\label{eq:plus1}\tag{$\plus.1$}\\
    &\suc(x) + y \simeq  \suc(x + y)&\label{eq:plus2}\tag{$\plus.2$}
\end{align}
\end{minipage}\begin{minipage}{.5\linewidth}
\begin{align}
    &\nil \appif x \simeq x&\label{eq:app1}\tag{$\append.1$}\\
    &\cons(x,y) \appif z \simeq \cons(x,y \appif z)&\label{eq:app2}\tag{$\append.2$}\\[.7em]
    &\lenif{\nil}\simeq\zero&\label{eq:len1}\tag{$\len.1$}\\
    &\lenif{\cons(x,y)}\simeq\suc(\lenif{y})&\label{eq:len2}\tag{$\len.2$}\\[.7em]
    &g(x)\appif f(y)\simeq f(x)\appif g(y)&\label{eq:ax1}\tag{$\mtt{ax}.1$}\\
    &\lenif{x}+\lenif{f(y)}\simeq\lenif{f(x)}+\lenif{y}&\label{eq:ax2}\tag{$\mtt{ax}.2$}\\
    &\lenif{f(g(x))}\simeq\suc(\lenif{x})&\label{eq:ax3}\tag{$\mtt{ax}.3$}
\end{align}
\end{minipage}
\begin{minipage}{.5\linewidth}
\smallskip\hrule\smallskip
{\bf Negated conjecture} --- $\neg\Conj$
\smallskip\hrule
{\small \begin{equation}\lenif{g(d)\appif f(c)}\not\simeq\suc(\lenif{c}+\lenif{d})\label{eq:negconj}\tag{$\mtt{nConj}$}\end{equation}}
\end{minipage}\hspace{.01\linewidth}\begin{minipage}{.49\linewidth}
\smallskip\hrule\smallskip
{\bf Auxiliary lemma} --- \Lem
\smallskip\hrule
{\small \begin{equation}\forall x,y.\lenif{x\appif y}\simeq\lenif{y}+\lenif{x}\label{eq:lemma}\tag{$\mtt{lemma}$}\end{equation}}
\end{minipage}
\medskip
\caption{Motivating example, conjecturing that the first-order formula \Conj is implied by first-order axioms \Axioms.\label{fig:defs}}
\end{figure*}
\vspace{-1.7em}
\section{Motivating Example}
\label{motivating-example}
We motivate our work using Figure~\ref{fig:defs}. The data types \nat and \lst---corresponding to the term algebras of natural numbers and lists in first-order logic---are defined inductively using constructors as given by the first-order formulas \eqref{eq:nat1}--\eqref{eq:list3}.
Moreover, Figure~\ref{fig:defs} defines the recursive functions for the addition of natural numbers ($+$), list append ($\appif$) and list length ($\lenif{.}$), encoded by \eqref{eq:plus1}--\eqref{eq:len2}, and declares the uninterpreted functions $f$ and $g$. 
We use infix notation for the symbols $+$, $\appif$ and $\lenif{.}$. Axioms \eqref{eq:ax1}--\eqref{eq:ax3} define the behaviour of $f$ and $g$. 
All properties in Figure~\ref{fig:defs} are implicitly universally-quantified.

Suppose we are trying to prove that the axioms (denoted \Axioms) in Figure~\ref{fig:defs} imply the following first-order formula: 
\begin{equation}\forall x,y. \lenif{g(y)\appif f(x)}\simeq \suc(\lenif{x}+\lenif{y})\label{eq:conj}\tag{\Conj}\end{equation}
Proving \eqref{eq:conj} in classical first-order logic can be reduced to establishing the unsatisfiability of the negation of~\eqref{eq:conj} together with the axioms of Figure~\ref{fig:defs}. That is, we prove  unsatisfiability of \eqref{eq:negconj} together with the axioms~\eqref{eq:nat1}--\eqref{eq:ax3}; here,  \eqref{eq:negconj} is the  negated and Skolemized form, of \eqref{eq:conj}, using the Skolem (list) functions  $c,d$. 
While the axioms imply formula~\eqref{eq:conj} in the theory of lists and natural numbers, $\Axioms\land\eqref{eq:negconj}$ is not first-order unsatisfiable.
Showing unsatisfiability requires an additional first-order axiom over lists, in particular the {\it auxiliary lemma} \Lem of Figure~\ref{fig:defs}.
With this lemma, $\Axioms\land\eqref{eq:negconj}\land\Lem$ is unsatisfiable and validity of \Conj follows.

Let us make two key observations. First, \Lem is a stronger property than \Conj with respect to \Axioms; from \eqref{eq:negconj}, the negation of an instance of \Lem (denoted $\neg\Lem$ by slight abuse of notation) can be derived via rewriting with equal terms. Second, \Lem is not a first-order consequence of \Axioms, but it is valid with respect to \Axioms in the background theory of lists and natural numbers, a fact that can be shown by induction. We make use of these two observations to synthesize and use \Lem in the proof of \Conj as follows:
\begin{enumerate}[label=(\roman*),leftmargin=1.7em]
\item We derive $\neg\Lem$ from $\Axioms\land\eqref{eq:negconj}$ by rewriting. Soundness of rewriting ensures that the unsatisfiability of $\Axioms\land\neg\Lem$ implies the unsatisfiability of $\Axioms\land\eqref{eq:negconj}$.
\item We refute $\Axioms\land\neg\Lem$ by instantiating a valid induction schema with $\Lem$ to obtain a valid first-order induction \emph{axiom}, which in conjunction with $\Axioms\land\neg\Lem$ is unsatisfiable, implying the unsatisfiability of $\Axioms\land\eqref{eq:negconj}$ and hence the claim of Figure~\ref{fig:defs}.
\end{enumerate}
To derive $\neg\Lem$ from $\Axioms\land\eqref{eq:negconj}$, we apply the following rewriting steps:
\begin{align}
&\text{-- rewrite \eqref{eq:negconj}}&&\text{with \eqref{eq:plus2} resulting in}&    \lenif{g(d)\appif f(c)}\not\simeq\suc(\lenif{c})+\lenif{d}\label{eq:negconj1}\\
&\text{-- rewrite \eqref{eq:negconj1}}&&\text{with \eqref{eq:ax3} resulting in}&
    \lenif{g(d)\appif f(c)}\not\simeq\lenif{f(g(c))}+\lenif{d}\label{eq:negconj2}\\
&\text{-- rewrite \eqref{eq:negconj2}}&&\text{with \eqref{eq:ax2} resulting in}&
    \lenif{g(d)\appif f(c)}\not\simeq\lenif{g(c)}+\lenif{f(d)}\label{eq:negconj3}\\
&\text{-- rewrite \eqref{eq:negconj3}}&&\text{with \eqref{eq:ax1} resulting in}&
    \lenif{f(d)\appif g(c)}\not\simeq\lenif{g(c)}+\lenif{f(d)}\label{eq:negconj4}
\end{align}
Notice that clause~\eqref{eq:negconj4} is the negation of \Lem, instantiated with $f(d)$ and $g(c)$. To refute $\Axioms\land\neg\Lem$, we \textit{conjecture} $\Lem$ to be proven by induction, by taking the negation of \eqref{eq:negconj4} and \textit{by generalizing over} the term $f(d)$. Hence, we instantiate the following \textit{second-order structural induction formula} over lists of natural numbers with $\Lem$:
\begin{equation}\label{eq:list-schema}
  \forall F.\big(\big(F(\nil)\land\forall x\in\nat,y\in\lst.(F(y)\implies F(\cons(x,y)))\big)\implies \forall z\in\lst.F(z)\big)
\end{equation}
Showing the first-order unsatisfiability of this induction axiom in conjunction with $\Axioms\land\neg\Lem$ requires no further rewriting\footnote{See Appendix A for details}.\medskip

\noindent
Note that the above reasoning actually proves \Lem \textit{in addition} to \Conj, but it comes with the following two main challenges for proving Figure~\ref{fig:defs}:
\begin{enumerate}[label={{\bf (C\arabic*)}},wide=0em,leftmargin=0em]
\item \label{C1} use rewriting with equalities to derive $\neg\Lem$ from $\Axioms\land\eqref{eq:negconj}$ (Section~\ref{sec:rewriting});
\item \label{C2} combine inductive reasoning with first-order reasoning to refute $\Axioms\land\neg\Lem$ (Section~\ref{sec:induction}).
\end{enumerate}
\noindent
For tackling challenge~\ref{C1}, we use \textit{rewriting} inferences to rewrite equal terms and generate auxiliary lemmas.
However, such proof steps cannot always be performed with the ubiquitous (ordered) superposition inferences. Let us use a  Knuth-Bendix simplification ordering (KBO)  $\succ$~\cite{NieuwenhuisRubio:HandbookAR:paramodulation:2001} parameterized by a constant weight function and the precedence $\gg$:
$$d\gg c\gg g\gg f\gg(\lenif{.})\gg(\appif)\gg(+)\gg\cons\gg\nil\gg\suc\gg\zero$$
The ordering $\succ$ cannot orient the equalities of~\eqref{eq:plus2} and \eqref{eq:ax3} right-to-left so that clause~\eqref{eq:negconj2} could be derived by rewriting. Addressing such obstacles, we introduce an extension of the superposition calculus (Section~\ref{sec:rewriting}) to enable generating auxiliary lemmas during saturation. Our extension solves challenge~\ref{C1} and provides an efficient reasoning backend for challenge~\ref{C2}.

\section{Theoretical Background}

We assume familiarity with many-sorted first-order logic with equality.
Variables are denoted with $x$, $y$, $z$,  terms with $s$, $t$, $u$, $w$, $l$, $r$, all possibly with indices. 
A term is \textit{ground} if it contains no variables. 
We use the standard logical connectives $\neg$, $\lor$, $\land$, $\rightarrow$ and $\leftrightarrow$,
and quantifiers $\forall$ and $\exists$. 
A \textit{literal} is an atom or its negation.
The literal $\overline{L}$ denotes the complement of literal $L$. 
A disjunction of literals is a \textit{clause}. We denote clauses by $C, D$ and reserve the symbol $\square$ for the \textit{empty clause}
that is logically equivalent to $\bot$. We refer to the \textit{clausal normal form} of a formula $F$ by $\cnf(F)$. We assume that $\cnf$  preserves satisfiability, i.e. $F$ is satisfiable iff $\cnf(F)$ is satisfiable.
We use $\simeq$ to denote equality and write $\bowtie$ for either $\simeq$ or $\not\simeq$.

An \textit{expression $E$}  is a term, literal, clause or formula. 
We write $E[s]_p$ to state that the expression $E$ contains some distinguished occurrence of the term
$s$ at some position $p$. We might simply write $E[s]$ if the position $p$ is not relevant. Further, $E[s\mapsto t]$  denotes that this occurrence of $s$ is replaced with $t$; when $s$ is clear from the context, we simply write $E[t]$. We say that $t$ is a \textit{subterm} of $s[t]$, denoted by $t\trianglelefteq s[t]$; and a \textit{strict subterm} if additionally $t\neq s[t]$, denoted by $t\triangleleft s[t]$.
A \textit{substitution} is a mapping from variables to terms. We denote substitutions by $\theta$, $\sigma$, $\rho$, $\mu$, $\eta$. A substitution $\theta$ is a \textit{unifier} of two terms $s$ and $t$ if $s\theta= t\theta$, and is a \textit{most general unifier} (denoted $\mgu(s,t)$) if for every unifier $\eta$ of $s$ and $t$, there exists a substitution $\mu$ s.t. $\eta=\theta\mu$.

A \emph{position} is a finite sequence of positive integers. The \emph{root position} is the empty sequence, denoted by $\epsilon$. Let $p$ and $q$ be positions. The \emph{concatenation} of $p$ and $q$ is denoted by $pq$. We say that $p$ is \emph{above} $q$ if there exists a position $r$ such that $pr=q$, denoted by $p\leqslant q$. We say that $p$ and $q$ are \textit{parallel}, denoted by $p\parallel q$, if $p\nleqslant q$ and $q\nleqslant p$. We say that $p$ is \emph{to the left} of $q$, denoted by $p<_l q$, if there are positive integers $i$ and $j$, positions $r$, $p^\prime$ and $q^\prime$ such that $p=rip^\prime$, $q=rjq^\prime$ and $i<j$.

Let $\to$ be a binary relation. The \textit{inverse} of $\to$ is denoted by $\leftarrow$. The \textit{reflexive-transitive closure} of $\to$ is denoted by $\to^*$. A binary relation $\to$ over the set of terms is a \textit{rewrite relation} if (i) $l\to r \Rightarrow l\theta\to r\theta$ and (ii) $l\to r \Rightarrow s[l]\to s[r]$ for any term $l$, $r$, $s$ and substitution $\theta$.
A \textit{rewrite ordering} is a strict rewrite relation. A \textit{reduction ordering} is a well-founded rewrite ordering. In this paper we consider reduction orderings total on ground terms. Such orderings satisfy $s \triangleright t \Rightarrow s \succ t$ and are also called \emph{simplification orderings}.

\subsection{Saturation-Based Theorem Proving}\label{sec:sat:sup}
We briefly introduce saturation-based proof search in first-order theorem proving.
For  details, we refer to~\cite{CAV13,GettingSaturated22}.
The majority of first-order theorem provers work with clauses, rather than arbitrary formulas. Let $S = \mathcal{A}\cup\{\neg G\}$ be a set of clauses including assumptions $\mathcal{A}$ and the clausified negation $\neg G$ of a goal $G$. Given $S$, first-order
provers {\it saturate} $S$ by computing all logical consequences of $S$ with respect to a sound inference system $\mathcal{I}$. This process is called \textit{saturation}. An inference system $\mathcal{I}$ is a set of inference rules of the form
\[
\frac{C_1\quad\ldots\quad C_n}{C},
\]
where $C_1,\ldots, C_n$ are the \emph{premises} and $C$ is the \emph{conclusion} of the inference. We also write $C_1,...,C_n\vdash_\mathcal{I}C$ to denote an inference in $\mathcal{I}$; as $\mathcal{I}$ is sound, this also means that $C$ is a logical consequence of $C_1,\quad,C_n$.  We denote that $\mathcal{I}$ derives clause $D$ from clauses $\mathcal{C}$ with $\mathcal{C}\vdash^*_\mathcal{I}D$. If the the saturated set of $S$ contains the empty clause $\square$, the original set $S$ of clauses is unsatisfiable, implying validity of $\mathcal{A}\rightarrow G$; in this case, we established a \emph{refutation} of $\neg G$ from $\mathcal{A}$.

\begin{figure}[tb]
\begin{tabular}{c p{.1\linewidth} l}
\multirow{3}{.45\linewidth}{
\centering
\AxiomC{$\underline{s[u]\bowtie t}\lor C$}
\AxiomC{$\underline{l\simeq r}\lor D$}
\LeftLabel{(\mtt{Sup})}
\BinaryInfC{$(s[r]\bowtie t\lor C \lor D)\theta$}
\DisplayProof}
&
\multirow{3}{*}{where} & (1) $u$ is not a variable,\\
& & (2) $\theta=\mgu(l,u)$,\\
& & (3) $r\theta\not\succeq l\theta$ and $t\theta\not\succeq s\theta$,\\[.3em]

\multirow{2}{.45\linewidth}{
\centering
\AxiomC{$\underline{s\not\simeq t}\lor C$}
\LeftLabel{(\mtt{EqRes})}
\UnaryInfC{$C\theta$}
\DisplayProof}
&
\multirow{2}{*}{where} & \multirow{2}{*}{$\theta=\mgu(s,t)$,}\\[1.8em]

\multirow{2}{.45\linewidth}{
\centering
\AxiomC{$\underline{s\simeq t}\lor \underline{u\simeq w} \lor C$}
\LeftLabel{(\mtt{EqFac})}
\UnaryInfC{$(s\simeq t\lor t\not\simeq w \lor C)\theta$}
\DisplayProof}
&
\multirow{2}{*}{where} & (1) $\theta=\mgu(s,u)$,\\
& & (2) $t\theta\not\succeq s\theta$ and $w\theta\not\succeq t\theta$.\\[.5em]
\end{tabular}
\caption{The superposition calculus \Sup{} for first-order logic with equality\label{fig:sup}.}
\end{figure}

Completeness and efficiency of saturation-based reasoning relies
on selecting and adding clauses from/to $S$ using the inference system $\mathcal{I}$. 
To constrain the inference system,
some first-order provers use simplification orderings
on terms. Simplification orderings are extended to
orderings over literals and clauses using the bag extension of the ordering; for simplicity, we write $\succ$  both for  the term ordering and its clause ordering extensions. Given an ordering $\succ$, a clause $C$ is \textit{redundant} with respect to a set $S$ of clauses if there exists a subset $S'$ of $S$ such that $S'$ implies $C$ and is smaller than $\{C\}$, i.e. $S'\implies C$ and $\{C\}\succ S'$.

The \textit{superposition calculus}, denoted \Sup{} and given in Figure~\ref{fig:sup}, is the most common inference system used by saturation-based first-order theorem provers~\cite{NieuwenhuisRubio:HandbookAR:paramodulation:2001}. We assume a literal selection function satisfying the standard condition on $\succ$ and underline selected literals in \Sup{} inferences. The \Sup{} calculus is \textit{sound} and \textit{refutationally complete}: for any unsatisfiable formula $\neg G$, the empty clause $\square$ can be derived as a logical consequence of $\neg G$.

\subsection{Inductive Reasoning in Saturation}\label{sec:ind}
Inductive reasoning has recently been  embedded in  saturation-based theorem proving~\cite{Cruanes17,CADE19}, by extending $\Sup$ with a new inference rule. More precisely, we introduce a family of induction inference rules parameterized by a second-order formula $G$ with exactly one free second-order variable $F$: the formula over which induction should be applied. Moreover, we restrict inductions to a set of terms $\mathcal{I}nd(\mathcal{T})\subseteq\mathcal{T}$ where $\mathcal{T}$ is the set of all terms. Then, the inference rules are of the following form:
\begin{center}
\begin{tabular}{c p{.1\linewidth} p{.45\linewidth}}
\multirow{3}{*}{
\AxiomC{$\overline{L}[t] \lor C$}
\LeftLabel{$(\Ind_G)$}
\UnaryInfC{$\cnf(\neg G[L[x]] \lor C)$}
\DisplayProof
}
&
\multirow{3}{*}{where}
&
(1) $L[t]$ is ground and $t\in\mathcal{I}nd(\mathcal{T})$,\\
& & (2) $\forall F.(G[F]\rightarrow \forall x.F[x])$ is a valid\\
& & \hspace{1.7em}second-order \textit{induction schema}.
\end{tabular}
\end{center}
By an \emph{induction axiom} we refer to an instance of a valid induction schema. When performing an $\Ind_G$ inference, the induction schema $\forall F.(G[F]\rightarrow \forall x.F[x])$ is said to be \textit{applied on} the clause $\overline{L}[t]\lor C$, or alternatively speaking $\overline{L}[t]\lor C$ is \textit{inducted upon}; in addition, we also say that we 
\textit{induct on term} $t$ in clause $\overline{L}[t]\lor C$ with induction schema $\forall F.(G[F]\rightarrow \forall x.F[x])$. For example, using the schema~\eqref{eq:list-schema}, we parameterize the $\Ind_G$ schema with $G:=F[\nil]\land\forall x,y.(F[y]\to F[\cons(x,y)])$ and obtain the $\Ind_G$ instance: 
\begin{center}
\begin{tabular}{c p{.05\linewidth} l}
\multirow{3}{*}{
\AxiomC{$\overline{L}[t] \lor C$}
\LeftLabel{$(\Ind_{G})$}
\UnaryInfC{$\begin{matrix}\overline{L}[\nil]\lor L[c_y] \lor C\\
\overline{L}[\nil]\lor \overline{L}[\cons(c_x,c_y)] \lor C\end{matrix}$}
\DisplayProof}
&
\multirow{3}{*}{where} &
(1) $L[t]$ is ground,\\
& & (2) $t\in\mathcal{I}nd(\mathcal{T})$ is of sort $\lst$,\\
& & (3) $c_x$ and $c_y$ are fresh Skolem symbols.\medskip\\
\end{tabular}
\end{center}
Note that the above $\Ind_G$ inference instance yields two clauses.

\section{Efficient Rewriting in Saturation}\label{sec:rewriting}

As motivated in Section~\ref{motivating-example}, rewriting derives clauses useful for auxiliary lemma generation that \Sup{} is not able to derive.
We therefore focus on rewriting variants captured by  the following inference rule:
\begin{center}
    \AxiomC{$C[l\theta]$}
    \AxiomC{$l\simeq r$}
    \LeftLabel{$(\Rw)$}
    \BinaryInfC{$C[r\theta]$}
    \DisplayProof
\end{center}
where $\theta$ is a substitution. 
We call an \Rw inference a \emph{downward rewrite} if $l\theta\succ r\theta$, and call an \Rw inference an \emph{upward rewrite} if $l\theta\prec r\theta$. 

We start by defining our base inference system, called the \textit{\textbf{Re}writing \textbf{C}alculus} (\ReC), as the calculus extending $\Sup$ with $\Rw$. In other words, we define \ReC to  consists of the inference rules of $\Sup\cup\{\Rw\}$. The refutational completeness of \ReC follows from the completeness of its subsystem \Sup.
 In addition to completeness, we consider the following property over inference systems, and in particular over \ReC. 
\begin{restatable}[Equational derivability (ED)]{definition}{edDef}
\label{def:ed}
Let $\theta$ be a substitution. 
An inference system $\mathcal{I}$ \emph{admits equational derivability (ED)} if, for any set of equations $\mathcal{C}$, equation $l\simeq r$ and clause $D[l\theta]$, if $\mathcal{C}\vdash^*_\mathcal{I} D[l\theta]$  then $\mathcal{C},l\simeq r\vdash^*_\mathcal{I} D[r\theta]$.
\end{restatable}
Equational derivability in Definition~\ref{def:ed} essentially expresses that an inference system can simulate the application of the \Rw rule, by some (possibly longer) derivation. This allows us to introduce and compare variants of \ReC, by imposing additional  rewriting constraints in \Rw. We state the following, straightforward result.
\begin{restatable}[\ReC--ED]{theorem}{recEgd}
The inference system \ReC admits ED.
\end{restatable}
In the sequel, we develop three improved variants of \ReC that admit ED, and thus  derive the same consequences with equations as \ReC does.

\subsection{Peak Elimination in \ReC}
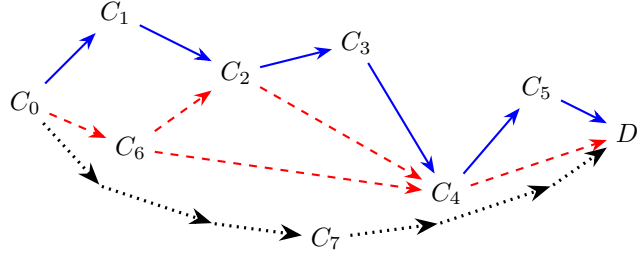
\begin{figure}[t]
\centering
\begin{tikzpicture}[n/.style={fill,circle,inner sep=2pt},outer sep=0pt]
    \node (s) at (0,0) {$C_0$};
    \node (s1) at (1.2,1.2) {$C_1$};
    \node (s1p) at (1.4,-0.6) {$C_6$};
    \node (s2) at (2.8,0.4) {$C_2$};
    \node (s3) at (4.4,0.8) {$C_3$};
    \node (s4) at (5.6,-1.2) {$C_4$};
    \node (s5) at (6.8,0.2) {$C_5$};
    \node (t) at (8,-0.4) {$D$};
    \node (sp) at (4,-1.8) {$C_7$};
    \node[inner sep=0pt] (sp1) at (1,-1.1) {};
    \node[inner sep=0pt] (sp2) at (2.5,-1.6) {};
    \node[inner sep=0pt] (sp3) at (5.5,-1.6) {};
    \node[inner sep=0pt] (sp4) at (7,-1.1) {};
    \draw[bluearrow] (s) -- (s1) {};
    \draw[bluearrow] (s1) -- (s2) {};
    \draw[bluearrow] (s2) -- (s3) {};
    \draw[bluearrow] (s3) -- (s4) {};
    \draw[bluearrow] (s4) -- (s5) {};
    \draw[bluearrow] (s5) -- (t) {};
    \draw[redarrow] (s) -- (s1p) {};
    \draw[redarrow] (s1p) -- (s2) {};
    \draw[redarrow] (s1p) -- (s4) {};
    \draw[redarrow] (s2) -- (s4) {};
    \draw[redarrow] (s4) -- (t) {};
    \draw[blackarrow] (s) -- (sp1) {};
    \draw[blackarrow] (sp1) -- (sp2) {};
    \draw[blackarrow] (sp2) -- (sp) {};
    \draw[blackarrow] (sp) -- (sp3) {};
    \draw[blackarrow] (sp3) -- (sp4) {};
    \draw[blackarrow] (sp4) -- (t) {};
\end{tikzpicture}
\caption{Possible rewrite sequences from a ground clause $C_0$ to a ground clause $D$. The total order between clauses is $C_1\succ C_3\succ C_2\succ C_5\succ C_0\succ D\succ C_6\succ C_4\succ C_7$, as also visualized by the vertical  alignment of clauses.}
\label{fig:par-paths}
\end{figure}
Let $\mathcal{C}$ be a satisfiable set of clauses. Suppose there is some ground clause $D$ that triggers the generation of a necessary inductive axiom, and suppose $D$ can be derived from $\mathcal{C}$ via rewrites with equations in $\mathcal{C}$. Hence our goal is to derive $D$. In Figure \ref{fig:par-paths}, we show\footnote{similarly to~\cite{goal_oriented_completion}.} possible ways to derive $D$  from a ground clause $C_0\in\mathcal{C}$ using equations in $\mathcal{C}$. Arrows of Figure \ref{fig:par-paths} point in the direction of deduction. Assume that all clauses in Figure~\ref{fig:par-paths} are ground; using a  total simplification ordering $\succ$ over ground clauses, we order clauses in 
Figure~\ref{fig:par-paths} as given by their vertical alignment in Figure~\ref{fig:par-paths}.
Therefore, an arrow going vertically upwards (resp. downwards) in Figure~\ref{fig:par-paths} corresponds to an upward (resp. downward) rewrite variant of \Rw. We use three different arrows in Figure~\ref{fig:par-paths}, corresponding to paths available at different saturation  steps (iterations) while    saturating $\mathcal{C}$:
\begin{enumerate}[label=(\arabic*),leftmargin=1.7em]
\item Arrows \bluearrow{2em} designate a path which is possible in a certain iteration $i$ during saturation, that is with equations available at iteration $i$.
\item Arrows \redarrow{2em} correspond to paths in later iterations than $i$ but not necessarily at the end of the saturation process.
\item Arrows \blackarrow{2em} correspond to the ``ideal path'' at the end of the saturation process, that is, when the equations are transformed into a set of equations corresponding to a complete (non-overlapping and terminating) rewrite system.
\end{enumerate}
As shown by the many  rewriting steps of 
Figure~\ref{fig:par-paths}, choosing a  path between $C_0$ and $D$ is not trivial. For example, using arrows \bluearrow{2em}, we may derive $D$ from $C_0$ in iteration $i$ already, but in principle we have to exhaustively apply rewrites in all ``directions'', resulting in many duplicate clauses. A different strategy is to wait until saturation end, in which case using arrows \blackarrow{2em} we rewrite $C_0$ into its normal form $C_7$, and then from $C_7$ we reach $D$ only by upward rewrites. However, saturation may never terminate, for example  in the presence of associativity and commutativity (AC) axioms.

Another option is to find a path of a specific form during saturation, such as the paths designated by \redarrow{2em} arrows in Figure~\ref{fig:par-paths}. We propose  to avoid so-called \textit{peaks} during saturation, where a peak comes with an upward rewrite followed by a downward rewrite. That is, upward rewrites followed by downward rewrites should be avoided. Depending on the positions in which the  upward and downward rewrites happen, the following two possibilities occur:  

\begin{enumerate}[label=(\roman*),leftmargin=1.7em]
\item If the positions of 
upward and downward rewrites
are \textit{parallel}, the two rewrites can be simply flipped. For example,  the path $C_0\bluearrow{2em}C_1\bluearrow{2em}C_2$ of  Figure \ref{fig:par-paths} is replaced by the path $C_0\redarrow{2em}C_6\redarrow{2em}C_2$.

\item If the positions are \textit{overlapping}, there is a superposition between the two rewriting equations of the peak. This superposition  inference generates an equation that ``cuts'' the peak, giving a one-step rewrite alternative instead of two rewrites. For example, the peak $C_4\bluearrow{2em} C_5\bluearrow{2em} D$ can be replaced by $C_4\redarrow{2em} D$. Note that sometimes multiple superpositions have to be performed before the path can be continued, e.g. the (double)
peak $C_6\redarrow{2em} C_2\bluearrow{2em} C_3\bluearrow{2em} C_4$ needs two superpositions to be eliminated, and performed simply as $C_6\redarrow{2em} C_4$.
\end{enumerate}
\begin{figure}[t]
\centering
\begin{tikzpicture}[outer sep=0]

    \node[] (s0) at (0,0) {\small \eqref{eq:negconj} $\lenif{g(d)\appif f(c)}\not\simeq\suc(\lenif{c}+\lenif{d})$};
    \node[] (s1) at (0.5,1.5) {\small \eqref{eq:negconj1} $\lenif{g(d)\appif f(c)}\not\simeq\suc(\lenif{c})+\lenif{d}$};
    \node[] (s2) at (1,3.2) {\small \eqref{eq:negconj2} $\lenif{g(d)\appif f(c)}\not\simeq\lenif{f(g(c))}+\lenif{d}$};
    \node[] (s3) at (6,2.2) {\small \eqref{eq:negconj3} $\lenif{g(d)\appif f(c)}\not\simeq\lenif{g(c)}+\lenif{f(d)}$};
    \node[anchor=east,xshift=0pt] (s4) at (11.5,1.2) {\small \eqref{eq:negconj4} $\lenif{f(d)\appif g(c)}\not\simeq\lenif{g(c)}+\lenif{f(d)}$};
    \node[] (s5) at (4.8,-1) {\small$\lenif{f(d)\appif g(c)}\not\simeq\suc(\lenif{c}+\lenif{d})$};
    \node[] (s6) at (5.2,.5) {\small$\lenif{f(d)\appif g(c)}\not\simeq\suc(\lenif{c})+\lenif{d}$};
    \draw[bluearrow] (s0) -- (s1) node[midway,xshift=-1.4em] {\small\ref*{eq:plus2}};
    \draw[bluearrow] (s1) -- (s2) node[midway,xshift=-.8em,yshift=.5em] {\small\ref*{eq:ax3}};
    \draw[bluearrow] (s2) -- (s3) node[midway,yshift=.5em] {\small\ref*{eq:ax2}};
    \draw[bluearrow] (s3) -- (s4.north) node[midway,yshift=.5em] {\small\ref*{eq:ax1}};
    \draw[redarrow] (s1) -- (s3) node[midway,yshift=-.7em] {\small\ref*{eq:sup1}};
    \draw[redarrow] (s0) -- (s5)  node[midway,yshift=-.6em] {\small\ref*{eq:ax1}};
    \draw[redarrow] (s5) -- (s6)  node[midway,xshift=1.4em] {\small\ref*{eq:plus2}};
    \draw[redarrow] (s6) -- (s4.south)  node[midway,yshift=-.6em] {\small\ref*{eq:sup1}};
    \draw[redarrow] (s1) -- (s6)  node[midway,yshift=-.6em] {\small\ref*{eq:ax1}};
\end{tikzpicture}
\label{fig:ex-paths}
\caption{Possible rewrite sequences to derive clause~\eqref{eq:negconj4} from~\eqref{eq:negconj} in the example of Section~\ref{motivating-example}.}
\end{figure}
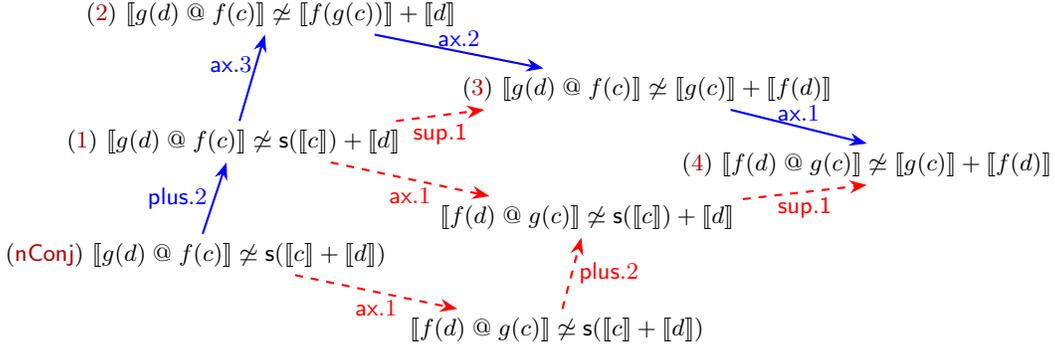
To avoid  peaks in saturation, we distinguish clauses  resulting from upward rewrites (annotated as $\urw{C}$) from other clauses (annotated as  $\drw{C}$). We use the notation $\erw{C}$ to denote either of these. We might leave clauses without annotation if this information is not relevant in the context. We split the \Rw inference into two components, resulting in the following inferences:\medskip
\begin{center}
\begin{minipage}{.5\linewidth}
\centering
\AxiomC{$\drw{C[l\theta]}$}
\AxiomC{$\drw{l\simeq r}$}
\LeftLabel{$(\Rw_\downarrow)$}
\BinaryInfC{$\drw{C[r\theta]}$}
\DisplayProof\quad
where $l\theta\npreceq r\theta$,
\end{minipage}\begin{minipage}{.5\linewidth}
\AxiomC{$\erw{C[l\theta]}$}
\AxiomC{$\drw{l\simeq r}$}
\LeftLabel{$(\Rw_\uparrow)$}
\BinaryInfC{$\urw{C[r\theta]}$}
\DisplayProof
\quad where $l\theta\nsucceq r\theta$.
\end{minipage}
\end{center}
\medskip
We denote our \emph{\ReC variant for  avoiding peaks in saturation}  by $\ReC_\lor$, and define  $\ReC_\lor$ to consist of the inference rules of $\Sup\cup\{\Rw_\downarrow,\Rw_\uparrow\}$. 
\begin{remark}
Note that the $\Rw_\downarrow$ and $\Rw_\uparrow$ rules both allow rewriting with incomparable equations. The reason for this is that disallowing rewrites with incomparable equations after upward rewrites violates ED in some cases\footnote{see Appendix B for details}.
\end{remark}
\begin{example}
Solid blue lines (\bluearrow{2em}) in  Figure~\ref{fig:ex-paths} show the sequence of rewrite steps to reach clause~\eqref{eq:negconj4} from clause~\eqref{eq:negconj} within the  motivating example of Section~\ref{motivating-example}, when using \ReC. 

Alternatively, we can perform the following \redarrow{2em} steps with $\ReC_\lor$. A superposition into clause~\eqref{eq:ax2} with \eqref{eq:ax3} results in
\begin{equation}
\lenif{g(x)}+\lenif{f(y)}\simeq\suc(\lenif{x})+\lenif{y}.\label{eq:sup1}\tag{\mtt{sup.1}}
\end{equation}
Using clause~\eqref{eq:sup1}, we eliminate the peak through clause~\eqref{eq:negconj2}, and directly derive clause~\eqref{eq:negconj3} from clause~\eqref{eq:negconj1} in $\ReC_\lor$. 
Note  that we can switch the order of rewrites using clauses~\eqref{eq:ax1} and \eqref{eq:plus2}; and similarly the order of rewrites using clauses~\eqref{eq:ax1} and \eqref{eq:sup1}.
We thus obtain the $\ReC_\lor$ derivation  of clause~\eqref{eq:negconj4} via rewriting clause~\eqref{eq:negconj}  with~\eqref{eq:ax1}, then with~\eqref{eq:plus2}, and finally with~\eqref{eq:sup1}.\qed

\end{example}
We conclude $\ReC_\lor$ with the following result.
\begin{restatable}[$\ReC_\lor$--ED]{theorem}{ugdRwDownarrow}
\label{thm:rw_v_ugd}
The inference system $\ReC_\lor$ admits ED.
\end{restatable}

\subsection{Diamond Elimination in \ReC}

Note that rewrites in parallel positions can be performed in any order. If the rewriting is performed in all possible orders, this leads to a large number of duplicated clauses. In this section, we restrict \ReC and $\ReC_\lor$ to eliminate this effect, while preserving equational derivability from  Definition~\ref{def:ed}.

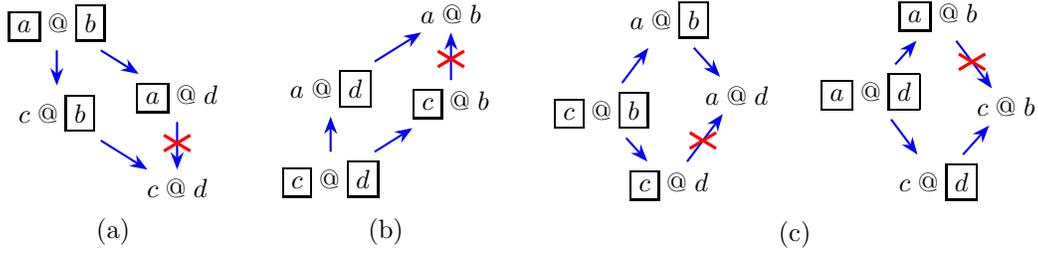
\begin{figure}[t]
\begin{minipage}{.25\linewidth}
\centering
\begin{tikzpicture}
    \node[] (s0) at (0,0) {$\boxed{a}\appif\boxed{b}$};
    \node[] (s1) at (0,-1.2) {$c\appif \boxed{b}$};
    \node[] (s2) at (1.6,-1.0) {$\boxed{a}\appif d$};
    \node[] (s3) at (1.6,-2.2) {$c\appif d$};
    \draw[downa,blue] (s0) -- (s1) {};
    \draw[downa,blue] (s0) -- (s2) {};
    \draw[downa,blue] (s1) -- (s3) {};
    \draw[downa,blue] (s2) -- (s3) {};
    \draw [very thick, red,-] ($(s2)!0.5!(s3)-(5pt,3pt)$) -- 
            ($(s2)!0.5!(s3)+(5pt,3pt)$);
    \draw [very thick, red,-] ($(s2)!0.5!(s3)-(5pt,-3pt)$) -- 
            ($(s2)!0.5!(s3)+(5pt,-3pt)$);
\end{tikzpicture}\\
(a)
\end{minipage}\begin{minipage}{.25\linewidth}
\centering
\begin{tikzpicture}
    \node[] (s0) at (1.6,0) {$a\appif b$};
    \node[] (s1) at (1.6,-1.2) {$\boxed{c}\appif b$};
    \node[] (s2) at (0,-1.0) {$a\appif \boxed{d}$};
    \node[] (s3) at (0,-2.2) {$\boxed{c}\appif \boxed{d}$};
    \draw[downa,blue] (s3) -- (s1) {};
    \draw[downa,blue] (s3) -- (s2) {};
    \draw[downa,blue] (s1) -- (s0) {};
    \draw[downa,blue] (s2) -- (s0) {};
    \draw [very thick, red,-] ($(s1)!0.5!(s0)-(5pt,3pt)$) -- 
            ($(s1)!0.5!(s0)+(5pt,3pt)$);
    \draw [very thick, red,-] ($(s1)!0.5!(s0)-(5pt,-3pt)$) -- 
            ($(s1)!0.5!(s0)+(5pt,-3pt)$);
\end{tikzpicture}\\
(b)
\end{minipage}\begin{minipage}{.5\linewidth}
\centering
\begin{tikzpicture}
    \node[] (s0) at (0,0) {$a\appif \boxed{b}$};
    \node[] (s1) at (-.9,-1.2) {$\boxed{c}\appif \boxed{b}$};
    \node[] (s2) at (.9,-1.0) {$a\appif d$};
    \node[] (s3) at (0,-2.2) {$\boxed{c}\appif d$};
    \draw[downa,blue] (s1) -- (s3) {};
    \draw[downa,blue] (s1) -- (s0) {};
    \draw[downa,blue] (s3) -- (s2) {};
    \draw[downa,blue] (s0) -- (s2) {};
    \draw [very thick, red,-] ($(s3)!0.5!(s2)-(5pt,3pt)$) -- 
            ($(s3)!0.5!(s2)+(5pt,3pt)$);
    \draw [very thick, red,-] ($(s3)!0.5!(s2)-(5pt,-3pt)$) -- 
            ($(s3)!0.5!(s2)+(5pt,-3pt)$);
\end{tikzpicture}
\quad
\begin{tikzpicture}
    \node[] (s0) at (0,0) {$\boxed{a}\appif b$};
    \node[] (s1) at (.9,-1.2) {$c\appif b$};
    \node[] (s2) at (-.9,-1.0) {$\boxed{a}\appif \boxed{d}$};
    \node[] (s3) at (0,-2.2) {$c\appif \boxed{d}$};
    \draw[downa,blue] (s2) -- (s3) {};
    \draw[downa,blue] (s2) -- (s0) {};
    \draw[downa,blue] (s3) -- (s1) {};
    \draw[downa,blue] (s0) -- (s1) {};
    \draw [very thick, red,-] ($(s0)!0.5!(s1)-(5pt,3pt)$) -- 
            ($(s0)!0.5!(s1)+(5pt,3pt)$);
    \draw [very thick, red,-] ($(s0)!0.5!(s1)-(5pt,-3pt)$) -- 
            ($(s0)!0.5!(s1)+(5pt,-3pt)$);
\end{tikzpicture}\\
(c)
\end{minipage}
\caption{Possible parallel rewrites in \ReC given equations $a\simeq c$, $b\simeq d$ with $a\succ c$ and $b\succ d$. Rewrites corresponding to crossed out arrows are not performed in the left-to-right order.\label{fig:diamonds}}
\end{figure}

Figure~\ref{fig:diamonds} illustrates rewriting possibilities for two parallel rewrites with \ReC. Depending on the direction of the rewrites, there are three possibilities, that is, three ``diamonds'': Figure~\ref{fig:diamonds}(a) shows two downward rewrites; Figure~\ref{fig:diamonds}(b) illustrates two upward rewrites; while Figure~\ref{fig:diamonds}(c) shows one downward and one upward rewrite. Note that Figure~\ref{fig:diamonds}(c) contains two subcases, one where the upward rewrite happens in the left position and one where it happens in the right position. We denote the positions to be rewritten with boxes, e.g. $\boxed{a}$. 
To generate all terms in these diamonds without duplicating any terms, we follow the tradition of reduction strategies in programming languages, for example \textit{leftmost-outermost} (also \textit{call-by-need}) and \textit{leftmost-innermost} (also \textit{call-by-value}) strategies~\cite{TermRewritingAndAllThat}. We choose a \textit{left-to-right} rewriting order, that is, we cannot perform a rewrite to the left of the previous rewrite. In Figure~\ref{fig:diamonds}, the skipped rewrites are crossed out in red.

Figure~\ref{fig:diamonds}(a)  and Figure~\ref{fig:diamonds}(b) are the same in $\ReC_\lor$. However,  $\ReC_\lor$ avoids duplication of Figure~\ref{fig:diamonds}(c) in the first place (recall the parallel positions in Figure~\ref{fig:par-paths}). Therefore, in the case of $\ReC_\lor$, we only apply the left-to-right order in the case of multiple consecutive \mtt{Rw_\downarrow} (resp. \mtt{Rw_\uparrow}) inferences. Towards this, we associate with each clause $C$ a position $p$ where the previous rewrite was performed. We denote such clauses by $_pC$. Our modified \Rw rule for avoiding duplicated diamonds is:\medskip
\begin{center}
    \AxiomC{$_pC[l\theta]_q$}
    \AxiomC{$l\simeq r$}
    \LeftLabel{$(\Rw^\to)$}
    \BinaryInfC{$_qC[r\theta]_q$}
    \DisplayProof
    \quad where $q\nless_l p$.
\end{center}
Our \ReC variant for \emph{avoiding duplicated diamonds during rewritings} is denoted 
by $\ReC^\to$ and is defined to be $\Sup\cup\{\Rw^\to\}$. We state the following result.
\begin{restatable}[$\ReC^\to$--ED]{theorem}{recToEgd}
\label{thm:rec_to_egd}
The inference system $\ReC^\to$ admits ED.
\end{restatable}
Finally, we define a \ReC variant that \emph{combines peak-elimination with  left-to-right rewriting orders}. We denote this \ReC variant by $\ReC^\to_\lor$. Here, we enforce the left-to-right order separately between downward and upward rewrites, as captured via the following \Rw variants: 
\begin{center}
\begin{tabular}{c p{.1\linewidth} l}
\multirow{2}{*}{\AxiomC{$\drw{_pC[l\theta]_q}$}
\AxiomC{$\drw{l\simeq r}$}
\LeftLabel{$(\Rw^\to_\downarrow)$}
\BinaryInfC{$\drw{_qC[r\theta]_q}$}
\DisplayProof}
& \multirow{2}{*}{where} & (1) $l\theta\npreceq r\theta$,\\
& & (2) $q\nless_l p$,\\[2em]
\multirow{2}{*}{
\AxiomC{$\erw{_pC[l\theta]_q}$}
\AxiomC{$\drw{l\simeq r}$}
\LeftLabel{$(\Rw^\to_\uparrow)$}
\BinaryInfC{$\urw{_qC[r\theta]_q}$}
\DisplayProof}
& \multirow{2}{*}{where} & (1) $l\theta\nsucceq r\theta$,\\
& & (2) $\erw{_pC[l\theta]}=\drw{_pC[l\theta]}$ or $q\nless_l p$.
\end{tabular}
\end{center}
\medskip
Our inference system $\ReC^\to_\lor$ is defined as $\Sup\cup\{\Rw^\to_\downarrow,\Rw^\to_\uparrow\}$ and has the following property. 
\begin{restatable}[$\ReC^\to_\lor$--ED]{theorem}{recToLorEgd}
\label{thm:rec_to_lor_egd}
The inference system $\ReC^\to_\lor$ admits ED.
\end{restatable}

\section{Redundancy and Induction}\label{sec:induction}

As mentioned in Section~\ref{sec:introduction}, the \Sup calculus tries to derive and retain as few clauses as possible without losing (refutational)  completeness. Within the  \ReC calculus, as well as within its three refinements $\ReC_\lor$, $\ReC^\to$ and $\ReC_\lor^\to$, however, we not only derive more consequences, but we  prevent their simplifications, resulting in less efficient reasoning than via $\Sup$. The situation  gets even worse when the prolific $\Ind_G$ rule is used to enable inductive reasoning.

As a remedy, this section integrates  redundancy elimination within our \ReC calculi extended with $\Ind_G$ inferences. Our main goal is to be as efficient as possible without losing inductive proofs. This includes, for example, preserving the ED property for our calculi extended with $\Ind_G$. We  introduce sufficient criteria to skip induction inferences in \ReC, and  weaken the ED restriction to avoid deriving  useless clauses for first-order and inductive reasoning.


We  identify induction inferences that can be omitted without losing proofs and provide efficient ways to detect such (redundant) inferences. 
\begin{remark}\label{ex:trivial-inductions}
Note that constructor-based induction schemas give rise to a few optimisations:
\begin{enumerate}[label=(\arabic*),leftmargin=1.7em]
\item Inducting on $t$ in $L[t]\lor C$ where $t$ has zero occurrences is possible, but using constructor-based induction schemas only results in tautological clauses $L\lor \neg L\lor C$ and clauses with duplicate literals $L\lor L\lor C$. We thus omit inducting   on $t$ in $L[t]\lor C$ where $t$ has zero occurrences. 

\item Inducting on base constructors, such as $\zero$ and $\nil$, only give weaker forms of the same clauses. For example, inducting upon  $\nil$ in $L[\nil]\lor C$ yields clauses of the form $L[\nil]\lor C'\lor C$. We thus omit induction on base costructors. 
\end{enumerate}
\end{remark}
While the inductive inferences described in Remark~\ref{ex:trivial-inductions} can easily be detected, this is not the case with more complex but useless inductive inferences, as shown in the next example. 
\begin{example}\label{ex:redundant-induction}
Consider the $\Ind_G$ inferences on term $c$ in clauses~\eqref{eq:negconj} and~\eqref{eq:negconj2}, respectively. The first $\Ind_G$ inference yields clauses
\begin{align*}
    &\lenif{g(d)\appif f(\nil)}\not\simeq\suc(\lenif{\nil}+\lenif{d})\lor\lenif{g(d)\appif f(c_2)}\simeq\suc(\lenif{c_2}+\lenif{d})\\
    &\lenif{g(d)\appif f(\nil)}\not\simeq\suc(\lenif{\nil}+\lenif{d})\lor\lenif{g(d)\appif f(\cons(c_1,c_2))}\not\simeq\suc(\lenif{\cons(c_1,c_2)}+\lenif{d})
\end{align*}
where $c_1$ and $c_2$ are fresh Skolem constants. The second $\Ind_G$ inference yields clauses
\begin{align*}
    &\lenif{g(d)\appif f(\nil)}\not\simeq\lenif{f(g(\nil))}+\lenif{d}\lor\lenif{g(d)\appif f(c_4)}\simeq\lenif{f(g(c_4))}+\lenif{d}\\
    &\lenif{g(d)\appif f(\nil)}\not\simeq\lenif{f(g(\nil))}+\lenif{d}\lor\lenif{g(d)\appif f(\cons(c_3,c_4))}\not\simeq\lenif{f(g(\cons(c_3,c_4)))}+\lenif{d} 
\end{align*}
where $c_3$ and $c_4$ are fresh Skolem constants. The two clause sets are equisatisfiable w.r.t. axioms~\eqref{eq:nat1}--\eqref{eq:ax3} and yield the same consequences, hence it is sufficient to retain only one. This is unfortunately hard to detect due to the different sets of Skolem constants $c_1,c_2$ and $c_3,c_4$. For example, simplifying the induction formulas and checking them for equivalence before clausification takes considerable effort.\qed
\end{example}
For detecting  redundancies similar to Example~\ref{ex:redundant-induction}, 
we characterize  redundant induction inferences of interest and  introduce sufficient conditions to efficiently detect them. 
\begin{definition}[Redundant $\Ind_G$ inference]
Let $C$ be a clause and $F$ a formula. The $\Ind_G$ inference $C \vdash \cnf(F)$ is \textit{redundant} w.r.t. a set of equations $E$ and a clause $C'$, if $C\succ C'$ and there is a formula $F'$ and an $\Ind_G$ inference $C'\vdash\cnf(F')$ s.t. $F$ and $F'$ are equivalent modulo rewriting with $E$.
\end{definition}
It is easy to see that only non-redundant $\Ind_G$ inferences need to be performed to retain equational consequences and first-order refutations. The following two lemmas show sufficient conditions to efficiently check for redundant induction inferences.
\begin{restatable}[Redundant $\Ind_G$ --  Condition I]{lemma}{redundantIndInferenceI}
\label{reducibility-lemma}
Let $l\simeq r$ and $\overline{L}[t]\lor C$ be clauses, $x$ a fresh variable, and $\mathcal{I}$ an inference system that admits ED. If there is a substitution $\theta$ s.t. $l\theta\triangleleft L[x]$ and $l\sigma\succ r\sigma$ where $\sigma=\theta\cdot\{x\mapsto t\}$, then the $\Ind_G$ inference
$$\overline{L}[t]\lor C\vdash\cnf(\neg G[L[x]]\lor C)$$
is redundant in $\mathcal{I}\cup\{\Ind_G\}$ w.r.t. the clauses $l\simeq r$ and $(\overline{L}[t])[l\sigma\mapsto r\sigma]\lor C$.
\end{restatable}

\begin{restatable}[Redundant $\Ind_G$ -- Condition II]{lemma}{redundantIndInferenceII}
\label{above-position-lemma}
Let $l\simeq r$ and $\overline{L}[t]\lor C$ be clauses, $x$ a fresh variable, and $\mathcal{I}$ an inference system that admits ED. If there is a substitution $\theta$ s.t. $l\theta\trianglelefteq t$ and $l\theta\succ r\theta$, then the $\Ind_G$ inference
$$\overline{L}[t]\lor C\vdash \cnf(\neg G[L[x]]\lor C)$$
is redundant in $\mathcal{I}\cup\{\Ind_G\}$ w.r.t. the clauses $l\simeq r$ and $\overline{L}[t[l\theta\mapsto r\theta]]\lor C$.
\end{restatable}
%
Lemmas~\ref{reducibility-lemma}--\ref{above-position-lemma} allow us to check for redundant $\Ind_G$ inferences similarly as performing demodulation, i.e. simplification by downward rewrites~\cite{SubsumptionDemodulation}. A consequence of these lemmas is that any non-redundant $\Ind_G$ inference on a clause that could be simplified by demodulation must induct on a subterm of a demodulatable term (i.e. a term that could be downward rewritten into a smaller term). The converse, however, that every subterm of a demodulatable term gives rise to a non-redundant $\Ind_G$ inference, does not hold in general.
\begin{example}\label{ex:subterms}
As shown in Example~\ref{ex:redundant-induction}, the $\Ind_G$ inference on clause~\eqref{eq:negconj2} using induction term $c$ is redundant w.r.t. clause~\eqref{eq:ax3} and~\eqref{eq:negconj} due to $\lenif{f(g(c))}\succ\suc(\lenif{c})$. The only non-redundant $\Ind_G$ inferences on clause~\eqref{eq:negconj2} w.r.t. clause~\eqref{eq:ax3} thus induct on the subterms $f(g(c))$ or $g(c)$ of $\lenif{f(g(c))}$.\qed
\end{example}
To control over which clauses induction should be triggered, we introduce the following notion for clauses that are not directly usable as premises for induction inferences.
\begin{definition}[Inductively redundant clause]
A clause is \textit{inductively redundant} if it is (first-order) redundant and all $\Ind_G$ inferences on it are redundant.
\end{definition}
The following example shows an inductively redundant clause.
\begin{example}
Continuing Example~\ref{ex:subterms}, the term $\lenif{g(d)\appif f(c)}$ in clause~\eqref{eq:negconj2} can be  demodulated, namely into $\lenif{f(d)\appif g(c)}$ by clause~\eqref{eq:ax1}. This makes inducting only on the subterms $g(d)$, $f(c)$ or $g(d)\appif f(c)$ non-redundant. As $\lenif{f(g(c))}$ and $\lenif{g(d)\appif f(c)}$ render all $\Ind_G$ inferences on the subterms of each other redundant, clause~\eqref{eq:negconj2} is inductively redundant.\qed
\end{example}
An inductively redundant clause is only necessary to preserve the ED property. Next, we show how some inductively redundant clauses can be avoided without losing ED. Towards this, we define so-called \emph{ineffective} equations.
\begin{definition}[Ineffective equation]
An equation $l\simeq r$ is \textit{ineffective} if $l\succ r$, each variable in $l$ has at most one occurrence and there is no strict non-variable subterm $s$ in $l$
s.t. $s\theta\in\mathcal{I}nd(\mathcal{T})$ for some substitution $\theta$. An equation is called \textit{effective} if it is not ineffective. We call an upward rewrite with an ineffective equation an \textit{ineffective rewrite}.
\end{definition}
The following example shows ineffective equations.
\begin{example}
As discussed in Example~\ref{ex:trivial-inductions}, base constructors such as $\zero$ and $\nil$ are not inducted upon when using only constructor-based induction schemas; hence,  $\zero,\nil\notin\mathcal{I}nd(\mathcal{T})$. The equations~\eqref{eq:plus1}, \eqref{eq:app1} and~\eqref{eq:len1} are therefore ineffective, since they are oriented left-to-right, their left-hand sides are linear, and none of the strict non-variable subterms in their left-hand sides are inducted upon.\qed
\end{example}
The following lemma proves that the result of an ineffective rewrite is inductively redundant.
\begin{restatable}[Redundancy of ineffective rewrites]{lemma}{redundantWeaklyUsefulRewrite}
\label{non-useful-equation-lemma}
Let $l\simeq r$ be an ineffective equation and $C[l\theta]$ a ground clause. If  $C[l\theta]\succ (l\simeq r)\theta$, the clause $C[l\theta]$ is inductively redundant in any inference system that admits ED. 
\end{restatable}
Consider a derivation of \Rw inferences from an inductively non-redundant clause $C$ into an inductively non-redundant conclusion $D$, where every intermediate clause in the derivation is inductively redundant. We may notice in such derivations that an ineffective rewrite is eventually followed by a rewrite that is not ineffective in an overlapping position. These rewrites can be performed together to avoid the intermediate inductively redundant clauses. To control and trigger such rewriting chains, we introduce the following \textit{chaining} inferences.
\begin{center}
\begin{tabular}{c l l}

\multirow{2}{*}{
\AxiomC{$C[l\theta]$}
\AxiomC{$l\simeq r$}
\LeftLabel{$(\mtt{CRw})$}
\BinaryInfC{$C[r\theta]$}
\DisplayProof}
& \multirow{2}{*}{where} & \multirow{2}{*}{$l\theta\nprec r\theta$ or $l\simeq r$ is effective,}\\[2em]

\multirow{3}{*}{
\AxiomC{$s[l']\simeq t$}
\AxiomC{$l\simeq r$}
\LeftLabel{$(\Cleft)$}
\BinaryInfC{$(s[r]\simeq t)\theta$}
\DisplayProof}
& \multirow{3}{*}{where} & (1) $\theta=\mgu(l,l')$,\\
& & (2) $s[l']\simeq t$ is ineffective,\\
& & (3) $l\theta\not\succeq r\theta$ and $l\simeq r$ is effective,\\[1em]

\multirow{3}{*}{
\AxiomC{$s\simeq t[l']$}
\AxiomC{$l\simeq r$}
\LeftLabel{$(\Cright)$}
\BinaryInfC{$(s\simeq t[r])\theta$}
\DisplayProof}
& \multirow{3}{*}{where} & (1) $\theta=\mgu(l,l')$,\\
& & (2) $l\simeq r$ is ineffective,\\
& & (3) $t[l']\theta\not\succeq s\theta$ and $s\simeq t[l']$ is effective.\\[1em]
\end{tabular}
\end{center}
The chaining inferences \Cleft and \Cright combine ineffective and effective equations together in new, effective equations. Further,  $\mtt{CRw}$ disallows ineffective rewrites for consequence generation. By using chaining inferences for efficient rewrites in saturation with induction, we  define the calculus $\CReC$ as $\Sup\cup\{\mtt{CRw},\Cleft,\Cright\}$. While  $\CReC$ does not admit ED, we note that if an inductively non-redundant clause $D$ is derivable via $\Rw$ in \ReC, then $D$ is also derivable using chaining inferences in $\CReC$. That is, inductive consequences are not lost in $\CReC$. The following theorem adjusts such a variant of ED to $\CReC$. 

\begin{restatable}[$\CReC$ derivability]{theorem}{recChaining}
\label{thm:chaining}
Let $D$ be an inductively non-redundant clause. If $\mathcal{C}\vdash_{\{\Rw\}}^\ast D$, then $\mathcal{C}\vdash_{\{\mtt{CRw},\Cleft,\Cright\}}^\ast D$.
\end{restatable}
The calculi $\CReC_\lor$, $\CReC^\to$ and $\CReC^\to_\lor$ are respectively the variants of $\ReC_\lor$, $\ReC^\to$ and $\ReC^\to_\lor$, when using $\Cleft$ and $\Cright$, and restricting the $\Rw$ variants in these calculi to be used with effective equations similarly as in $\CRw$.\footnote{see Appendix C} Derivability results for these calculi similar to Theorem~\ref{thm:chaining} are straightforward based on Theorems~\ref{thm:rw_v_ugd}--\ref{thm:chaining}.

\section{Evaluation}\label{sec:evaluation}
\paragraph{{\bf Implementation.}} 
We implement our calculi in the \vampire\footnote{\url{https://github.com/vprover/vampire/commit/16a38442515f8385}} prover. 
Our framework for equational consequence generation is controlled via the new option {\tt -grw} which has the following values: {\tt off} disables equational consequence generation and uses only the \Sup calculus; {\tt all} uses the calculus \ReC; {\tt up} uses $\ReC_\lor$; {\tt ltr} uses $\ReC^\to$; and {\tt up\_ltr} uses $\ReC_\lor^\to$. With the further option {\tt -mgrwd}, we  limit the maximum depth of rewrites for \Rw inference variants. The option {\tt -grwc} toggles the chaining inferences, that is, the use of \CReC variants. Finally, the option {\tt -indrc} controls the redundancy check for induction, by using Lemmas~\ref{reducibility-lemma}--\ref{above-position-lemma} to avoid performing redundant $\Ind_G$ inferences. 

{Additionally, we use the following heuristics to control consequence generation in saturation with induction. We apply rewriting in a goal-oriented manner, only allowing rewriting into conjectures and their subgoals. Moreover, to avoid useless clauses from rewriting between unrelated subgoals, we disallow rewriting inferences which would introduce new Skolem constants into a clause. Finally, we avoid rewriting inferences which introduce variables into our conjectures, as these have to be instantiated before induction.}


\paragraph{{\bf Experimental setup.}} We run our experiments with the following option setup: {\tt -sa discount} to use   the 
 \textsc{Discount} saturation algorithm ~\cite{Discount};  {\tt -drc encompass} to enable  encompassment demodulation~\cite{Duarte2022}; and {\tt -thsq on} to control pure theory derivations~\cite{SplitQueues}. Experiments were run on computers with AMD Epyc 7502 2.5GHz processors and 1TB RAM, with each individual benchmark run given a single core. For inductive reasoning experiments, we used the UFDTLIA benchmark set from SMT-LIB~\cite{SMTLIB}, the TIP benchmark set~\cite{TIP} and the \vampire{} inductive benchmark set~\cite{CICM20}. We also used benchmarks from the UEQ division of TPTP~\cite{TPTP} to test first-order reasoning.

\paragraph{\bf Evaluation of inductive reasoning.} 
The first part of our experiments consisted of running \vampire{} on \emph{1266 inductive benchmarks} from the UFDTLIA, TIP and \vampire{} benchmark sets. We used a 60-second timeout and the options {\tt -ind struct -indoct on} to enable induction and generalisations over complex terms. We used two different simplification orderings: {\tt -to kbo} for  KBO ordering with constant weight and precedence determined  by the arity of symbols; {\tt -to lpo -sp occurrence} for the LPO ordering with a symbol precedence given by the declaration order. This LPO order is usually better at orienting recursive function axioms~\cite{FMCAD21}.


\begin{figure}[tb]
\setlength\tabcolsep{3.6pt}

\centering
\begin{tabular}{c|c||c|c|c|c||c|c|c|c|}
{\tt -indrc off} & $\Sup$ & $\ReC$ & $\ReC_\lor$ & $\ReC^\to$ & $\ReC^\to_\lor$ & $\CReC$ & $\CReC_\lor$ & $\CReC^\to$ & $\CReC^\to_\lor$\\
\hline
\hline
KBO & 270 & 290 & 290 & 290 & {\bf 291} & 302 & {\bf 305} & 302 & 302 \\
\hline
LPO & 290 & 318 & 320 & {\bf 322} & 320 & 331 & {\bf 332} & {\bf 332} & {\bf 332} \\
\end{tabular}\medskip

\begin{tabular}{c|c||c|c|c|c||c|c|c|c|}
{\tt -indrc on} & $\Sup$ & $\ReC$ & $\ReC_\lor$ & $\ReC^\to$ & $\ReC^\to_\lor$ & $\CReC$ & $\CReC_\lor$ & $\CReC^\to$ & $\CReC^\to_\lor$\\
\hline
\hline
KBO & 270 & 299 & 299 & {\bf 301} & 300 & 305 & 304 & {\bf 307} & {\bf 307}\\
\hline
LPO & 290 & 321 & 319 & {\bf 322} & 321 & {\bf 333} & {\bf 333} & 330 & 329\\
\end{tabular}

\caption{Comparison of the $\Sup$ calculus with variants of the $\ReC$ and $\CReC$ calculi, using  1266 inductive benchmarks. Redundant $\Ind_G$ inference detection is disabled (resp. enabled) in the top (resp. bottom) table with option {\tt -indrc off} (resp. {\tt -indrc on}). Maximum rewriting depth ({\tt -mgrwd}) is set to  3.}
\label{fig:induction-comparison}
\end{figure}

Our results are summarised in Figure~\ref{fig:induction-comparison}, showcasing that each \ReC and \CReC calculi variant performs significantly better than \Sup. 
Using the LPO ordering turned out to be  advantageous over the KBO ordering. Performance is further improved via detection of redundant $\Ind_G$ inferences and using chaining inferences via \CReC variants. Among each group, however, the differences in the number of solved benchmarks are minimal, and there is no calculus that is a clear winner in all configurations. Statistics reveal that redundant $\Ind_G$ inference detection used together with $\ReC$ was able to eliminate 79.3\% of the overall 390,657,294 $\Ind_G$ inferences, while redundant $\Ind_G$ inference detection in $\CReC$ variant runs detected 42.6\% of the overall 169,457,851 $\Ind_G$ inferences redundant. This suggests that both redundant $\Ind_G$ detection and chaining inferences in $\CReC$ variants are effective in keeping the search space small. In total, the configurations for the $\ReC$ and $\CReC$ variants solved 45 problems that no \Sup variant could solve. Based on these results, we conclude that rewriting in inductive reasoning significantly improves upon standard superposition.

Figure~\ref{fig:cactusplots} shows cactus plots~\cite{CactusPlots} within the LPO configuration of $\Sup$,    $\ReC$ and $\CReC$ variants, and $\ReC$ and $\CReC$ variants with redundant $\Ind_G$ inference detection. Each plot line lists the logarithm of the time needed (vertical axis) to solve a certain number of the benchmarks (horizontal axis) individually, for a particular configuration. The left diagram shows the entire plot, and the right diagram a magnified (and rescaled) plot above 180 problems. The baseline configuration \Sup is a bit faster than the $\ReC$ variants up to around 260 problems, and after that it only solves a few problems in the several seconds region. The $\ReC$ variants without redundant $\Ind_G$ detection are almost indistinguishable in the entire plot. The $\ReC$ calculus with redundant $\Ind_G$ detection is however better, while there is a greater gap between the $\Sup$ and $\ReC$ calculi and the two $\CReC$ variants. The $\CReC$ calculus with redundant $\Ind_G$ detection has the slowest growing curve, corresponding to the fastest solving times.


\paragraph{\bf Evaluation of first-order  reasoning.} We also measured how our calculi behave with pure first-order problems. In particular, we have run experiments on the UEQ division of TPTP using the CASC2019 portfolio mode of \vampire{} with a 300 seconds timeout ({\tt --mode portfolio -sched casc\_2019 -t 300}). While our calculi performed slightly worse than \vampire{} portfolio, we managed to solve a few unique and hard UEQ problems: 
\vampire{} without consequence generation could not solve: {\tt GRP664-12} (rating 0.96), {\tt COL066-1} (rating 0.79), {\tt LAT166-1} (rating 0.71), {\tt LAT156-1} (rating 0.71) and {\tt REL026-1} (rating 0.67). As such, our work is  useful not only for inductive, but also for first-order reasoning.

\begin{figure}[tb]

\begin{tikzpicture}[spy using outlines={rectangle, magnification=2, connect spies}]
\begin{axis}[
    width=.4\linewidth,
    ymode=log,
    y tick label style={/pgf/number format/1000 sep=\,},
    log ticks with fixed point,
    scale only axis,
    y tick label style={font=\small},
    x tick label style={font=\small},
    yticklabels={0s,0.01s,0.1s,1s,10s},
    ymin=0.005,
    ymax=60,
    xmin=0,
    xmax=336,
    legend pos=north west,
    legend style={font=\scriptsize},
    every axis plot/.append style={semithick,mark repeat=10},
]
\addplot[mark=x,color=black]
    table[x index=0, y index=1] {figure1/lpo/baseline-times.txt};
\addlegendentry{$\Sup$};
\addplot[mark=10-pointed star,color=red]
    table[x index=0,y index=1] {figure1/lpo/all-times.txt};
\addlegendentry{$\ReC$};
\addplot[mark=o,color=green]
    table[x index=0,y index=1] {figure1/lpo/up-times.txt};
\addlegendentry{$\ReC_\lor$};
\addplot[mark=square,color=orange]
    table[x index=0,y index=1] {figure1/lpo/ltr-times.txt};
\addlegendentry{$\ReC^\rightarrow$};
\addplot[mark=triangle,color=blue]
    table[x index=0,y index=1] {figure1/lpo/up-ltr-times.txt};
\addlegendentry{$\ReC^\rightarrow_\lor$};
\addplot[mark=pentagon,color=brown]
    table[x index=0,y index=1] {figure1/lpo/all-indrc-times.txt};
\addlegendentry{$\ReC$({\tt indrc})};
\addplot[mark=Mercedes star flipped,color=magenta]
    table[x index=0,y index=1] {figure1/lpo/all-grwc-times.txt};
\addlegendentry{$\CReC$};
\addplot[mark=diamond,color=teal]
    table[x index=0,y index=1] {figure1/lpo/all-grwc-indrc-times.txt};
\addlegendentry{$\CReC$({\tt indrc})};
\end{axis}
\begin{axis}[
    axis line style={thick,blue},
    scale only axis=true,
    at={(0.425\linewidth,0)},
    width=.5\linewidth,
    height=.347\linewidth,
    ymode=log,
    yticklabel=\empty,
    x tick label style={font=\small},
    scale only axis,
    ymin=0.02,
    ymax=60,
    xmin=180,
    xmax=336,
    every axis plot/.append style={semithick,mark repeat=3},
]
\addplot[mark=x,color=black]
    table[x index=0, y index=1] {figure1/lpo/baseline-times.txt};
\addplot[mark=10-pointed star,color=red]
    table[x index=0,y index=1] {figure1/lpo/all-times.txt};
\addplot[mark=o,color=green]
    table[x index=0,y index=1] {figure1/lpo/up-times.txt};
\addplot[mark=square,color=orange]
    table[x index=0,y index=1] {figure1/lpo/ltr-times.txt};
\addplot[mark=triangle,color=blue]
    table[x index=0,y index=1] {figure1/lpo/up-ltr-times.txt};
\addplot[mark=pentagon,color=brown]
    table[x index=0,y index=1] {figure1/lpo/all-indrc-times.txt};
\addplot[mark=diamond,color=teal]
    table[x index=0,y index=1] {figure1/lpo/all-grwc-indrc-times.txt};
\addplot[mark=Mercedes star flipped,color=magenta]
    table[x index=0,y index=1] {figure1/lpo/all-grwc-times.txt};
\end{axis}
\draw[draw=blue] (3.1,0.7) rectangle ++(2.7,4.3);
\draw[blue,densely dashed] (3.1,0.7) to (6.15,0);
\draw[blue,densely dashed] (5.8,0.7) to (6.15,0.65);
\draw[blue,densely dashed] (5.8,5) to (6.15,5.04);
\end{tikzpicture}

\caption{Plots of (logarithmic) time against number of problems that would be solved individually given that time limit. Selected configurations, all using LPO.}
\label{fig:cactusplots}
\end{figure}
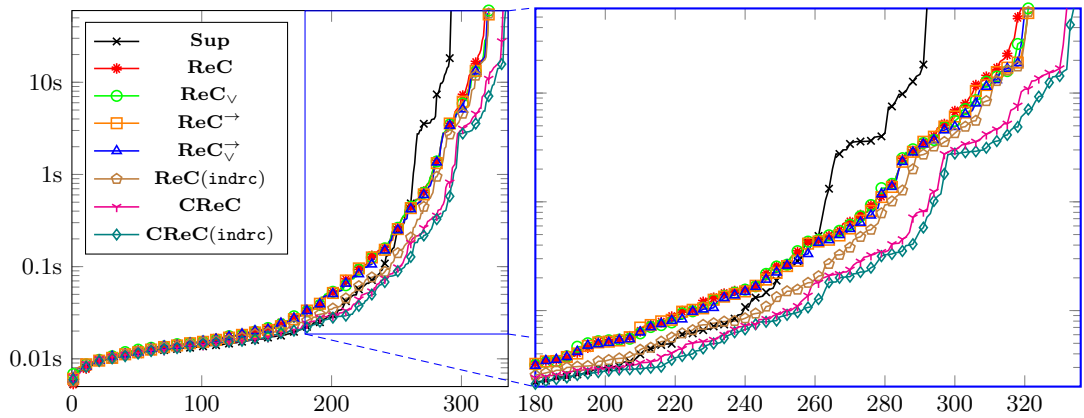

\vspace{-.5em}
\section{Related Work and Conclusion}
We improve the generation of equational consequences within saturation-based theorem proving extended with inductive reasoning. The generated consequences serve as auxiliary lemmas to be used for proving (inductive) goals.

While auxiliary lemmas might be provided by users in  interactive theorem proving ~\cite{IsabelleHol,sledgehammer}, saturation-based automated theorem provers, by design, do not  support user guidance during proof search. Automation of induction in saturation therefore 
implements inductive generalizations~\cite{CADE19,CADE21,Aubin,Cruanes17}, using failed proof attempts \cite{Rippling93} and specialized sound inferences with common patterns \cite{Wand17}. Our work extends these techniques by guiding proof search with auxiliary lemmas generated during proof search. Lemma generation is also exploited in theory exploration~\cite{HipSpec,ConjectureSynthesisForInductiveTheories}, without however imposing the relevance of   generated lemmas with respect to a given conjecture, which is the focus of our method. 

The use of equational theories and term rewriting in inductive reasoning has  been addressed in~\cite{ACL2}, by ensuring termination of function definitions, and in~\cite{Wand17,FMCAD21}, by using  
using completion procedures~\cite{ImplicitInductionInConditionalTheories,Rippling93} to interleave heuristic rewriting with theorem proving. Our work complements these approaches, by using rewriting-based equational reasoning for auxiliary lemma generation, and extending redundancy elimination in inductive reasoning. 
 Our experiments have shown a significant improvement in inductive reasoning and equational first-order reasoning, solving  45 new inductive problems and 5 hard TPTP problems when compared to standard superposition. 
 
 Integrating our method with proof assistants, e.g. Sledgehammer~\cite{sledgehammer}, is an interesting line of future work, with the aim of splitting goals into subgoals described by auxiliary lemmas.
Applying our  approach  to non-equational fragments, such as Horn formulas in equational logic~\cite{HornToEquational}, is another future challenge.

\paragraph{Acknowledgements.}
We thank Martin Suda for discussing ideas related to this work. This work was partially funded by 
the ERC Consolidator Grant ARTIST 101002685,
the TU Wien Doctoral College SecInt, and 
the FWF SFB project SpyCoDe F8504,  and the
WWTF ICT22-007 grant ForSmart.


\bibliographystyle{plain}
\bibliography{bibliography}

\newpage

\section*{Appendix A}

\subsection*{Proof of our motivating example}
{
\setlength\tabcolsep{3.6pt}
\begin{longtable}{l l r}
1.&$\zero\not\simeq\suc(x)$&[\ref{eq:nat1}]\\
2.&$\suc(x)\not\simeq\suc(y)\lor x\simeq y$&[\ref{eq:nat2}]\\
3.&$\nil\not\simeq\cons(x,y)$&[\ref{eq:list1}]\\
4.&$\cons(x,y)\not\simeq\cons(z,u)\lor x\simeq z$&[\ref{eq:list2}]\\
5.&$\cons(x,y)\not\simeq\cons(z,u)\lor y\simeq u$&[\ref{eq:list3}]\\
6.&$\zero + x \simeq x$&[\ref{eq:plus1}]\\
7.&$\suc(x) + y \simeq  \suc(x + y)$&[\ref{eq:plus2}]\\
8.&$\nil \appif x \simeq x$&[\ref{eq:app1}]\\
9.&$\cons(x,y) \appif z \simeq \cons(x,y \appif z)$&[\ref{eq:app2}]\\
10.&$\lenif{\nil}\simeq\zero$&[\ref{eq:len1}]\\
11.&$\lenif{\cons(x,y)}\simeq\suc(\lenif{y})$&[\ref{eq:len2}]\\
12.&$g(x)\appif f(y)\simeq f(x)\appif g(y)$&[\ref{eq:ax1}]\\
13.&$\lenif{x}+\lenif{f(y)}\simeq\lenif{f(x)}+\lenif{y}$&[\ref{eq:ax2}]\\
14.&$\lenif{f(g(x))}\simeq\suc(\lenif{x})$&[\ref{eq:ax3}]\\
15.&$\lenif{g(d)\appif f(c)}\not\simeq\suc(\lenif{c}+\lenif{d})$&[\ref{eq:negconj}]\\
16.&$\lenif{g(d)\appif f(c)}\not\simeq\suc(\lenif{c})+\lenif{d}$&[\Rw 15,7]\\
17.&$\lenif{g(d)\appif f(c)}\not\simeq\lenif{f(g(c))}+\lenif{d}$&[\Rw 16,14]\\
18.&$\lenif{g(d)\appif f(c)}\not\simeq\lenif{g(c)}+\lenif{f(d)}$&[\Rw 17,13]\\
19.&$\lenif{f(d)\appif g(c)}\not\simeq\lenif{g(c)}+\lenif{f(d)}$&[\Rw 18,12]\\
20.&$\lenif{\nil\appif g(c)}\not\simeq\lenif{g(c)}+\lenif{\nil}\lor\lenif{c_2\appif g(c)}\simeq\lenif{g(c)}+\lenif{c_2}$&[\Ind 19]\\
21.&$\lenif{\nil\appif g(c)}\not\simeq\lenif{g(c)}+\lenif{\nil}\lor\lenif{\cons(c_1,c_2)\appif g(c)}\not\simeq\lenif{g(c)}+\lenif{\cons(c_1,c_2)}$&[\Ind 19]\\
22.&$\lenif{g(c)}\not\simeq\lenif{g(c)}+\lenif{\nil}\lor\lenif{c_2\appif g(c)}\simeq\lenif{g(c)}+\lenif{c_2}$&[\mtt{Sup} 20,8]\\
23.&$\lenif{g(c)}\not\simeq\lenif{g(c)}+\zero\lor\lenif{c_2\appif g(c)}\simeq\lenif{g(c)}+\lenif{c_2}$&[\mtt{Sup} 22,10]\\
24.&$\zero\not\simeq\zero+\zero\lor c_3\simeq c_3+\zero\lor\lenif{c_2\appif g(c)}\simeq\lenif{g(c)}+\lenif{c_2}$&[\Ind 23]\\
25.&$\zero\not\simeq\zero+\zero\lor \suc(c_3)\not\simeq \suc(c_3)+\zero\lor\lenif{c_2\appif g(c)}\simeq\lenif{g(c)}+\lenif{c_2}$&[\Ind 23]\\
26.&$\zero\not\simeq\zero\lor c_3\simeq c_3+\zero\lor\lenif{c_2\appif g(c)}\simeq\lenif{g(c)}+\lenif{c_2}$&[\mtt{Sup} 24,6]\\
27.&$c_3\simeq c_3+\zero\lor\lenif{c_2\appif g(c)}\simeq\lenif{g(c)}+\lenif{c_2}$&[\mtt{EqRes} 26]\\
28.&$\zero\not\simeq\zero\lor \suc(c_3)\not\simeq \suc(c_3)+\zero\lor\lenif{c_2\appif g(c)}\simeq\lenif{g(c)}+\lenif{c_2}$&[\mtt{Sup} 25,6]\\
29.&$\suc(c_3)\not\simeq \suc(c_3)+\zero\lor\lenif{c_2\appif g(c)}\simeq\lenif{g(c)}+\lenif{c_2}$&[\mtt{EqRes} 28]\\
30.&$\suc(c_3)\not\simeq \suc(c_3+\zero)\lor\lenif{c_2\appif g(c)}\simeq\lenif{g(c)}+\lenif{c_2}$&[\mtt{Sup} 29,7]\\
31.&$\suc(c_3)\not\simeq \suc(c_3)\lor\lenif{c_2\appif g(c)}\simeq\lenif{g(c)}+\lenif{c_2}$&[\mtt{Sup} 30,27]\\
32.&$\lenif{c_2\appif g(c)}\simeq\lenif{g(c)}+\lenif{c_2}$&[\mtt{EqRes} 31]\\

33.&$\lenif{g(c)}\not\simeq\lenif{g(c)}+\lenif{\nil}\lor\lenif{\cons(c_1,c_2)\appif g(c)}\not\simeq\lenif{g(c)}+\lenif{\cons(c_1,c_2)}$&[\mtt{Sup} 21,8]\\
34.&$\lenif{g(c)}\not\simeq\lenif{g(c)}+\zero\lor\lenif{\cons(c_1,c_2)\appif g(c)}\not\simeq\lenif{g(c)}+\lenif{\cons(c_1,c_2)}$&[\mtt{Sup} 33,10]\\
35.&$\zero\not\simeq\zero+\zero\lor c_4\simeq c_4+\zero\lor\lenif{\cons(c_1,c_2)\appif g(c)}\not\simeq\lenif{g(c)}+\lenif{\cons(c_1,c_2)}$&[\Ind 34]\\
36.&$\zero\not\simeq\zero+\zero\lor \suc(c_4)\not\simeq \suc(c_4)+\zero\lor\lenif{\cons(c_1,c_2)\appif g(c)}\not\simeq\lenif{g(c)}+\lenif{\cons(c_1,c_2)}$&[\Ind 34]\\
37.&$\zero\not\simeq\zero\lor c_4\simeq c_4+\zero\lor\lenif{\cons(c_1,c_2)\appif g(c)}\not\simeq\lenif{g(c)}+\lenif{\cons(c_1,c_2)}$&[\mtt{Sup} 35,6]\\
38.&$c_4\simeq c_4+\zero\lor\lenif{\cons(c_1,c_2)\appif g(c)}\not\simeq\lenif{g(c)}+\lenif{\cons(c_1,c_2)}$&[\mtt{EqRes} 37]\\
39.&$\zero\not\simeq\zero\lor \suc(c_4)\not\simeq \suc(c_4)+\zero\lor\lenif{\cons(c_1,c_2)\appif g(c)}\not\simeq\lenif{g(c)}+\lenif{\cons(c_1,c_2)}$&[\mtt{Sup} 36,6]\\
40.&$\suc(c_4)\not\simeq \suc(c_4)+\zero\lor\lenif{\cons(c_1,c_2)\appif g(c)}\not\simeq\lenif{g(c)}+\lenif{\cons(c_1,c_2)}$&[\mtt{EqRes} 39]\\
41.&$\suc(c_4)\not\simeq \suc(c_4+\zero)\lor\lenif{\cons(c_1,c_2)\appif g(c)}\not\simeq\lenif{g(c)}+\lenif{\cons(c_1,c_2)}$&[\mtt{Sup} 40,7]\\
42.&$\suc(c_4)\not\simeq \suc(c_4)\lor\lenif{\cons(c_1,c_2)\appif g(c)}\not\simeq\lenif{g(c)}+\lenif{\cons(c_1,c_2)}$&[\mtt{Sup} 41,27]\\
43.&$\lenif{\cons(c_1,c_2)\appif g(c)}\not\simeq\lenif{g(c)}+\lenif{\cons(c_1,c_2)}$&[\mtt{EqRes} 42]\\
44.&$\lenif{\cons(c_1,c_2\appif g(c))}\not\simeq\lenif{g(c)}+\lenif{\cons(c_1,c_2)}$&[\mtt{Sup} 43,9]\\
45.&$\suc(\lenif{c_2\appif g(c)})\not\simeq\lenif{g(c)}+\lenif{\cons(c_1,c_2)}$&[\mtt{Sup} 44,11]\\
46.&$\suc(\lenif{g(c)}+\lenif{c_2})\not\simeq\lenif{g(c)}+\lenif{\cons(c_1,c_2)}$&[\mtt{Sup} 45,32]\\
47.&$\suc(\lenif{g(c)}+\lenif{c_2})\not\simeq\lenif{g(c)}+\suc(\lenif{c_2})$&[\mtt{Sup} 46,11]\\
48.&$\suc(\zero+\lenif{c_2})\not\simeq\zero+\suc(\lenif{c_2})\lor\suc(c_5+\lenif{c_2})\simeq c_5+\suc(\lenif{c_2})$&[\Ind 47]\\
49.&$\suc(\zero+\lenif{c_2})\not\simeq\zero+\suc(\lenif{c_2})\lor\suc(\suc(c_5)+\lenif{c_2})\not\simeq\suc(c_5)+\suc(\lenif{c_2})$&[\Ind 47]\\
50.&$\suc(\lenif{c_2})\not\simeq\zero+\suc(\lenif{c_2})\lor\suc(c_5+\lenif{c_2})\simeq c_5+\suc(\lenif{c_2})$&[\mtt{Sup} 48,6]\\
51.&$\suc(\lenif{c_2})\not\simeq\suc(\lenif{c_2})\lor\suc(c_5+\lenif{c_2})\simeq c_5+\suc(\lenif{c_2})$&[\mtt{Sup} 50,6]\\
52.&$\suc(c_5+\lenif{c_2})\simeq c_5+\suc(\lenif{c_2})$&[\mtt{EqRes} 51]\\
53.&$\suc(\lenif{c_2})\not\simeq\zero+\suc(\lenif{c_2})\lor\suc(\suc(c_5)+\lenif{c_2})\not\simeq\suc(c_5)+\suc(\lenif{c_2})$&[\mtt{Sup} 49,6]\\
54.&$\suc(\lenif{c_2})\not\simeq\suc(\lenif{c_2})\lor\suc(\suc(c_5)+\lenif{c_2})\not\simeq\suc(c_5)+\suc(\lenif{c_2})$&[\mtt{Sup} 53,6]\\
55.&$\suc(\suc(c_5)+\lenif{c_2})\not\simeq\suc(c_5)+\suc(\lenif{c_2})$&[\mtt{EqRes} 54]\\
56.&$\suc(\suc(c_5+\lenif{c_2}))\not\simeq\suc(c_5)+\suc(\lenif{c_2})$&[\mtt{Sup} 55,7]\\
57.&$\suc(\suc(c_5+\lenif{c_2}))\not\simeq\suc(c_5+\suc(\lenif{c_2}))$&[\mtt{Sup} 56,7]\\
58.&$\suc(c_5+\suc(\lenif{c_2}))\not\simeq\suc(c_5+\suc(\lenif{c_2}))$&[\mtt{Sup} 57,52]\\
59.&$\square$&[\mtt{EqRes} 58]\\
\end{longtable}}

\newpage
\section*{Appendix B}
\subsection*{Rewriting calculi -- variants of \ReC}

\edDef*

\subsection*{The \ReC calculus}

\begin{center}
\begin{tabular}{c p{.1\linewidth} l}
\multirow{3}{*}{
\AxiomC{$\underline{s[u]\bowtie t}\lor C$}
\AxiomC{$\underline{l\simeq r}\lor D$}
\LeftLabel{(\mtt{Sup})}
\BinaryInfC{$(s[r]\bowtie t\lor C \lor D)\theta$}
\DisplayProof}
&
\multirow{3}{*}{where} & (1) $u$ is not a variable,\\
& & (2) $\theta=\mgu(l,u)$,\\
& & (3) $r\theta\not\succeq l\theta$ and $t\theta\not\succeq s\theta$,\\
\\

\multirow{2}{*}{
\AxiomC{$\underline{s\not\simeq t}\lor C$}
\LeftLabel{(\mtt{EqRes})}
\UnaryInfC{$C\theta$}
\DisplayProof}
&
\multirow{2}{*}{where} & \multirow{2}{*}{$\theta=\mgu(s,t)$,}\\
\\
\\

\multirow{2}{*}{
\AxiomC{$\underline{s\simeq t}\lor \underline{u\simeq w} \lor C$}
\LeftLabel{(\mtt{EqFac})}
\UnaryInfC{$(s\simeq t\lor t\not\simeq w \lor C)\theta$}
\DisplayProof}
&
\multirow{2}{*}{where} & (1) $\theta=\mgu(s,u)$,\\
& & (2) $t\theta\not\succeq s\theta$ and $w\theta\not\succeq t\theta$,\\
\end{tabular}

\begin{prooftree}
\AxiomC{$C[l\theta]$}
\AxiomC{$l\simeq r$}
\LeftLabel{$(\Rw)$}
\BinaryInfC{$C[r\theta]$}
\end{prooftree}
\medskip
\end{center}

\recEgd*
\begin{proof}
Let $\mathcal{C}$ be a set of equations and $l\simeq r$ an equation. By assumption we are given a derivation $\mathcal{C}\vdash^*_\ReC D[l\theta]$. We trivially extend it to the following \ReC-derivation:
\begin{prooftree}
\AxiomC{$\vdots$}
\noLine
\UnaryInfC{$D[l\theta]$}
\AxiomC{$l\simeq r$}
\LeftLabel{$(\Rw)$}
\BinaryInfC{$D[r\theta]$}
\end{prooftree}
\end{proof}

\subsection*{The $\ReC_\lor$ calculus}

\begin{center}
\begin{tabular}{c p{.1\linewidth} l}
\multirow{3}{*}{
\AxiomC{$\drw{\underline{s[u]\bowtie t}\lor C}$}
\AxiomC{$\drw{\underline{l\simeq r}\lor D}$}
\LeftLabel{(\mtt{Sup})}
\BinaryInfC{$\drw{(s[r]\bowtie t\lor C \lor D)\theta}$}
\DisplayProof}
&
\multirow{3}{*}{where} & (1) $u$ is not a variable,\\
& & (2) $\theta=\mgu(l,u)$,\\
& & (3) $r\theta\not\succeq l\theta$ and $t\theta\not\succeq s\theta$,\\
\\

\multirow{2}{*}{
\centering
\AxiomC{$\drw{\underline{s\not\simeq t}\lor C}$}
\LeftLabel{(\mtt{EqRes})}
\UnaryInfC{$\drw{C\theta}$}
\DisplayProof}
&
\multirow{2}{*}{where} & \multirow{2}{*}{$\theta=\mgu(s,t)$,}\\
\\
\\

\multirow{2}{.45\linewidth}{
\centering
\AxiomC{$\drw{\underline{s\simeq t}\lor \underline{u\simeq w} \lor C}$}
\LeftLabel{(\mtt{EqFac})}
\UnaryInfC{$\drw{(s\simeq t\lor t\not\simeq w \lor C)\theta}$}
\DisplayProof}
&
\multirow{2}{*}{where} & (1) $\theta=\mgu(s,u)$,\\
& & (2) $t\theta\not\succeq s\theta$ and $w\theta\not\succeq t\theta$,\\[1.8em]


\end{tabular}
\begin{minipage}{.5\linewidth}
\centering
\AxiomC{$\drw{C[l\theta]}$}
\AxiomC{$l\simeq r$}
\LeftLabel{$(\Rw_\downarrow)$}
\BinaryInfC{$\drw{C[r\theta]}$}
\DisplayProof\quad
where $l\theta\npreceq r\theta$,
\end{minipage}\begin{minipage}{.5\linewidth}
\AxiomC{$\erw{C[l\theta]}$}
\AxiomC{$l\simeq r$}
\LeftLabel{$(\Rw_\uparrow)$}
\BinaryInfC{$\urw{C[r\theta]}$}
\DisplayProof
\quad where $l\theta\nsucceq r\theta$.
\end{minipage}
\medskip
\end{center}



\begin{remark}
As noted earlier, we have to allow rewriting with incomparable equations in both $\Rw_\downarrow$ and $\Rw_\uparrow$ inferences. We show an example where a rewrite with an incomparable equation after an upward rewrite is needed.

Let $g(x,x)\simeq x$, $f(x,y)\simeq f(y,x)$ and $p(f(x,y))$ be clauses. Deriving $p(g(f(x,y),f(y,x)))$ can be done in two different ways:
\begin{enumerate}[label=(\roman*),leftmargin=1.7em]
\item
\begin{prooftree}
\AxiomC{$\drw{p(f(x,y))}$}
\AxiomC{$\drw{g(x,x)\simeq x}$}
\BinaryInfC{$\urw{p(g(f(x,y),f(x,y)))}$}
\AxiomC{$\drw{f(x,y)\simeq f(y,x)}$}
\BinaryInfC{$\urw{p(g(f(x,y),f(y,x)))}$}
\end{prooftree}

\item
\begin{prooftree}
\AxiomC{$\drw{p(f(x,y))}$}
\AxiomC{$\drw{f(x,y)\simeq f(y,x)}$}
\BinaryInfC{$\drw{p(f(y,x))}$}
\AxiomC{$\drw{g(x,x)\simeq x}$}
\BinaryInfC{$\urw{p(g(f(y,x),f(y,x)))}$}
\AxiomC{$\drw{f(x,y)\simeq f(y,x)}$}
\BinaryInfC{$\urw{p(g(f(x,y),f(y,x)))}$}
\end{prooftree}
\end{enumerate}
Notice that in both cases, we have to perform a rewrite with an incomparable equation after the upward rewrite.

We call an equation linear if both of its sides contain at most one occurrence of each variable. The equation $f(x,y)\simeq f(y,x)$ is linear, while $g(x,x)\simeq x$ is not. We can also avoid this issue if we restrict rewriting equations to be linear.
\end{remark}

\ugdRwDownarrow*
\begin{proof}
Let $\mathcal{C}$ be a set of equations and $l\simeq r$ an equation. We apply well-founded induction on $\succ$ for clauses. In particular, we assume contrary to the claim that we have a minimal (w.r.t. $\succ$) clause $\erw{D[l\theta]}$ such that $\mathcal{C}\vdash^*_{\ReC_\lor} \erw{D[l\theta]}$ but $\mathcal{C},l\simeq r\nvdash^*_{\ReC_\lor} \erw{D[r\theta]}$, violating equational derivability. Since this clause is a minimal clause (w.r.t. $\succ$) where we cannot perform such a derivation, we have the induction hypothesis that for all $\erw{D'[u\sigma]}\prec \erw{D[l\theta]}$ and all sets of equations $\mathcal{D}$ and equation $u\simeq w$, if $\mathcal{D}\vdash^*_{\ReC_\lor}\erw{D'[u\sigma]}$ then $\mathcal{D},u\simeq w\vdash^*_{\ReC_\lor}\erw{D'[w\sigma]}$. As we will see, the only problematic case in the proof is when we need to perform a downward rewrite after an upward rewrite. In this case, the clause before the upward rewrite is always smaller than $\erw{D[l\theta]}$, hence we can apply the induction hypothesis. We consider the following cases.
\begin{enumerate}[label=(\roman*),leftmargin=1.7em]

\item If $l\theta\nsucceq r\theta$, then we have the following derivation:
\begin{prooftree}
\AxiomC{$\vdots$}
\noLine
\UnaryInfC{$\erw{D[l\theta]}$}
\AxiomC{$l\simeq r$}
\LeftLabel{$(\Rw_\uparrow)$}
\BinaryInfC{$\urw{D[r\theta]}$}
\end{prooftree}
Hence we could derive $\erw{D[r\theta]}$, contradiction.

\item If $l\theta=r\theta$, then $\erw{D[l\theta]}=\erw{D[r\theta]}$, so we get the derivation trivially, contradiction.

\item If $l\theta\succ r\theta$ and $\erw{D[l\theta]}=\drw{D[l\theta]}$, then we have the following derivation:
\begin{prooftree}
\AxiomC{$\vdots$}
\noLine
\UnaryInfC{$\drw{D[l\theta]}$}
\AxiomC{$l\simeq r$}
\LeftLabel{$(\Rw_\downarrow)$}
\BinaryInfC{$\drw{D[r\theta]}$}
\end{prooftree}
Hence we could derive $\erw{D[r\theta]}$, contradiction.

\item If $l\theta\succ r\theta$ and $\erw{D[l\theta]}=\urw{D[l\theta]}$, then we have the following derivation by assumption:
\begin{prooftree}
\AxiomC{$\vdots$}
\noLine
\UnaryInfC{$\erw{D[u\sigma]}$}
\AxiomC{$u\simeq w$}
\LeftLabel{$(\Rw_\uparrow)$}
\BinaryInfC{$\urw{(D[w\sigma])[l\theta]}$}
\end{prooftree}
where $u\sigma\nsucceq w\sigma$. If $u\sigma\nprec w\sigma$, then $u\sigma$ and $w\sigma$ are incomparable. But then, there is also a derivation:
\begin{prooftree}
\AxiomC{$\vdots$}
\noLine
\UnaryInfC{$\erw{D[u\sigma]}$}
\AxiomC{$u\simeq w$}
\LeftLabel{$(\Rw_\downarrow)$}
\BinaryInfC{$\drw{(D[w\sigma])[l\theta]}$}
\end{prooftree}
This case is similar to Case (iii). Otherwise, we have $u\sigma\prec w\sigma$ and we consider the following subcases:
\begin{enumerate}[leftmargin=0em]

\item If $l\theta$ and $w\sigma$ are in parallel positions in $\urw{(D[w\sigma])[l\theta]}=\urw{D[w\sigma][l\theta]}$. Then, we have $\mathcal{C}\vdash^*_{\ReC_\lor}\urw{D[u\sigma][l\theta]}$ and since $u\sigma\prec w\sigma$, $\urw{D[u\sigma][l\theta]}\prec \urw{D[w\sigma][l\theta]}$, the claim holds for $\urw{D[u\sigma][l\theta]}$, in particular, from $\mathcal{C}\vdash^*_{\ReC_\lor}\urw{D[u\sigma][l\theta]}$ follows that and $\mathcal{C}\vdash^*_{\ReC_\lor}\erw{D[u\sigma][r\theta]}$. Then, we get the following derivation:
\begin{prooftree}
\AxiomC{$\vdots$}
\noLine
\UnaryInfC{$\erw{D[u\sigma][r\theta]}$}
\AxiomC{$u\simeq w$}
\LeftLabel{$(\Rw_\uparrow)$}
\BinaryInfC{$\urw{D[w\sigma][r\theta]}$}
\end{prooftree}
Hence we could derive $\erw{D[r\theta]}$, contradiction.

\item If $l\theta$ is a subterm of $w\sigma$ and is at or below a variable position in $w$, then there is some variable $x$ in $w$ s.t. $x\sigma=t[l\theta]$ for some term $t$. We have $\erw{D[u\sigma]}\prec \urw{D[w\sigma]}$ by assumption, hence from $\mathcal{C}\vdash^*_{\ReC_\lor}\erw{D[u\sigma]}$, (by applying the induction hypothesis several times, since $l\theta\succ r\theta$), we have a derivation $\mathcal{C},l\simeq r\vdash^*_{\ReC_\lor}\erw{D[u\sigma']}$ where for each $y$ in $u$
$$\sigma'(y)=\begin{cases}\sigma(y) & \text{if }x\neq y\\
t[r\theta] & \text{otherwise}\end{cases}$$
Note that this also works if $u$ does not contain $x$, because then $u\sigma=u\sigma'$ and we can use $\sigma'$ in the rewriting. Since $\erw{D[u\sigma']}\preceq \erw{D[u\sigma]}$, we can apply the induction hypothesis one more time to get from $\mathcal{C},l\simeq r\vdash^*_{\ReC_\lor}\erw{D[u\sigma']}$ the derivation $\mathcal{C},l\simeq r\vdash^*_{\ReC_\lor}\erw{D[w\sigma']}$ (note that $u\simeq w\in \mathcal{C}$ by assumption). Now we have $\erw{D[w\sigma']}=\erw{D[r\theta][r\theta]...[r\theta]}$ where some $r\theta$s should be $l\theta$s. By $l\theta\succ r\theta$, we restore these terms in the following derivation:
\begin{prooftree}
\AxiomC{$\vdots$}
\noLine
\UnaryInfC{$\erw{D[r\theta][r\theta]...[r\theta]}$}
\AxiomC{$l\simeq r$}
\LeftLabel{$(\Rw_\uparrow)$}
\BinaryInfC{$\urw{D[l\theta][r\theta]...[r\theta]}$}
\AxiomC{$l\simeq r$}
\LeftLabel{$(\Rw_\uparrow)$}
\BinaryInfC{$\ddots$}
\AxiomC{$l\simeq r$}
\LeftLabel{$(\Rw_\uparrow)$}
\BinaryInfC{$\urw{D[l\theta][l\theta]...[l\theta]}$}
\end{prooftree}
Hence we could derive $\erw{D[r\theta]}$, contradiction.

\item If $w\sigma$ is a subterm of $l\theta$ and is at or below a variable position in $l$, then there is some variable $x$ in $l$ s.t. $x\theta=t[w\sigma]$ for some term $t$. Consider the substitution $\theta'$ where for each variable $y$ in $l$, we have
$$\theta'(y)=\begin{cases}\theta(y) & \text{if }x\neq y\\
t[u\sigma] & \text{otherwise}\end{cases}$$
Since $\erw{D[u\sigma]}\prec \urw{(D[w\sigma])[l\theta]}$, and $u\sigma\prec w\sigma$, we can apply the induction hypothesis several times to get by $\mathcal{C}\vdash^*_{\ReC_\lor}\erw{D[u\sigma]}$ the derivation $\mathcal{C}\vdash^*_{\ReC_\lor}\erw{(D[u\sigma])[l\theta']}$. Then, by $\erw{(D[u\sigma])[l\theta']}\prec\urw{(D[w\sigma])[l\theta]}$ and by the induction hypothesis and $\mathcal{C}\vdash^*_{\ReC_\lor}\erw{(D[u\sigma])[l\theta']}$, we get a derivation $\mathcal{C},l\simeq r\vdash^*_{\ReC_\lor}\erw{(D[u\sigma])[r\theta']}$. Now we have the desired $\erw{(D[w\sigma])[r\theta]}$ except there are some $u\sigma$s which should be $w\sigma$s. By $u\sigma\prec w\sigma$, we restore these with the following derivation:
\begin{prooftree}
\AxiomC{$\vdots$}
\noLine
\UnaryInfC{$\erw{(D[u\sigma])[r\theta']}$}
\AxiomC{$u\simeq w$}
\LeftLabel{$(\Rw_\uparrow)$}
\BinaryInfC{$\ddots$}
\AxiomC{$u\simeq w$}
\LeftLabel{$(\Rw_\uparrow)$}
\BinaryInfC{$\urw{(D[w\sigma])[r\theta]}$}
\end{prooftree}
Hence we could derive $\erw{D[r\theta]}$, contradiction.

\item $l\theta$ is a subterm of $w\sigma$ in a position which is above a variable position in $w$. Let $l'$ be the term in $w$ in this position. Then, there is a superposition
\begin{prooftree}
\AxiomC{$w[l']\simeq u$}
\AxiomC{$l\simeq r$}
\BinaryInfC{$(w[r]\simeq u)\rho$}
\end{prooftree}
where $\rho=\mgu(l,l')$, $\theta=\rho\mu$ for some substitution $\mu$, and $\sigma=\rho\eta$ for some substitution $\eta$ (otherwise $l$ or $w$ would not match the clause $D$). Then, $\erw{D[u\sigma]}\prec \urw{(D[w\sigma])[l\theta]}$ so the induction hypothesis applies and we get from $\mathcal{C}\vdash^*_{\ReC_\lor} \erw{D[u\sigma]}$ the derivation $\mathcal{C},(w[r]\simeq u)\rho\vdash^*_{\ReC_\lor} \erw{(D[w\sigma])[r\theta]}$ and hence $\mathcal{C},l\simeq r\vdash^*_{\ReC_\lor} \erw{(D[w\sigma])[r\theta]}$.
Hence we could derive $\erw{D[r\theta]}$, contradiction.

\item $w\sigma$ is a subterm of $l\theta$ in a position which is above a variable position in $l$. Let $w'$ be the term in $l$ in this position. Then, there is a superposition
\begin{prooftree}
\AxiomC{$l[w']\simeq r$}
\AxiomC{$w\simeq u$}
\BinaryInfC{$(l[u]\simeq r)\rho$}
\end{prooftree}
where $\rho=\mgu(w,w')$, $\theta=\rho\mu$ for some substitution $\mu$, and $\sigma=\rho\eta$ for some substitution $\eta$ (otherwise $l$ or $w$ would not match the clause $D$). Then, $\erw{D[u\sigma]}\prec \urw{(D[w\sigma])[l\theta]}$ so the induction hypothesis applies and we get from $\mathcal{C}\vdash^*_{\ReC_\lor} \erw{D[u\sigma]}$ the derivation $\mathcal{C},(l[w]\simeq r)\rho\vdash^*_{\ReC_\lor} \erw{(D[w\sigma])[r\theta]}$ and hence $\mathcal{C},l\simeq r\vdash^*_{\ReC_\lor} \erw{(D[w\sigma])[r\theta]}$.
Hence we could derive $\erw{D[r\theta]}$, contradiction.

\end{enumerate}
We have covered all subcases of case (iv).
\end{enumerate}
We have covered all cases, which proves the claim.
%
%
%
\end{proof}

\subsection*{The $\ReC^\to$ calculus}

\begin{center}
\begin{tabular}{c p{.1\linewidth} l}
\multirow{3}{*}{
\AxiomC{$_p\underline{s[u]\bowtie t}\lor C$}
\AxiomC{$_q\underline{l\simeq r}\lor D$}
\LeftLabel{(\mtt{Sup})}
\BinaryInfC{$_\epsilon(s[r]\bowtie t\lor C \lor D)\theta$}
\DisplayProof}
&
\multirow{3}{*}{where} & (1) $u$ is not a variable,\\
& & (2) $\theta=\mgu(l,u)$,\\
& & (3) $r\theta\not\succeq l\theta$ and $t\theta\not\succeq s\theta$,\\
\\

\multirow{2}{*}{
\AxiomC{$_p\underline{s\not\simeq t}\lor C$}
\LeftLabel{(\mtt{EqRes})}
\UnaryInfC{$_\epsilon C\theta$}
\DisplayProof}
&
\multirow{2}{*}{where} & \multirow{2}{*}{$\theta=\mgu(s,t)$,}\\
\\
\\

\multirow{2}{*}{
\AxiomC{$_p\underline{s\simeq t}\lor \underline{u\simeq w} \lor C$}
\LeftLabel{(\mtt{EqFac})}
\UnaryInfC{$_\epsilon(s\simeq t\lor t\not\simeq w \lor C)\theta$}
\DisplayProof}
&
\multirow{2}{*}{where} & (1) $\theta=\mgu(s,u)$,\\
& & (2) $t\theta\not\succeq s\theta$ and $w\theta\not\succeq t\theta$,\\
\\

\multirow{2}{*}{\AxiomC{$_pC[l\theta]_q$}
\AxiomC{$_{p'}l\simeq r$}
\LeftLabel{$(\Rw^\to)$}
\BinaryInfC{$_qC[r\theta]_q$}
\DisplayProof} & \multirow{2}{*}{where} &\multirow{2}{*}{$q\nless_l p$.}\\
\medskip
\end{tabular}
\end{center}

\recToEgd*
\begin{proof}
We annotate each \ReC derivation $\Pi$ with a finite sequence of positions $S(\Pi)$ inductively as follows. If the derivation is empty, i.e. we have $\vdash D$, then $S(\vdash D)=\langle\rangle$ (the empty sequence). Otherwise, we have a derivation of the following form:
\begin{prooftree}
\AxiomC{$\vdots$}
\noLine
\UnaryInfC{$D[l\theta]_p$}
\AxiomC{$l\simeq r$}
\LeftLabel{$(\Rw)$}
\BinaryInfC{$D[r\theta]_p$}
\end{prooftree}
and we have $S(\mathcal{C}\vdash^* D[l\theta]_p)=\langle p_1,...,p_n\rangle$. Then, the sequence is $S(\mathcal{C},l\simeq r\vdash^* D[r\theta])=\langle p_1,...,p_n,p\rangle$. We define a partial ordering $\lessdot$ on finite sequences of positions as follows: $\langle p_1,...,p_n,p,q,q_1,...,q_m\rangle\lessdot\langle p_1,...,p_n,q,p,q_1,...,q_m\rangle$ if $p<_l q$. It is easy to see that $\lessdot$ is well-founded.

We show the claim by well-founded induction on $\lessdot$. Let $\mathcal{C}$ be a set of equations and $D$ a clause. Take a $\ReC$ derivation $\Pi=\mathcal{C}\vdash^* D$ s.t. $S(\Pi)$ is minimal w.r.t. $\lessdot$ but is not a $\ReC^\to$ derivation. This means that there are two consecutive inferences of the following form in the derivation:
\begin{prooftree}
\AxiomC{$\vdots$}
\noLine
\UnaryInfC{$_{p'}C[l\theta]_p[u\sigma]_q$}
\AxiomC{$u\simeq w$}
\BinaryInfC{$_qC[l\theta]_p[w\sigma]_q$}
\AxiomC{$l\simeq r$}
\BinaryInfC{$_pC[r\theta]_p[w\sigma]_q$}
\noLine
\UnaryInfC{$\vdots$}
\end{prooftree}
where $p<_l q$ and hence we have $S(\Pi)=\langle p_1,...,p_n,p',q,p,q_1,...,q_m\rangle$. Consider the derivation $\Pi'$:
\begin{prooftree}
\AxiomC{$\vdots$}
\noLine
\UnaryInfC{$_{p'}C[l\theta]_p[u\sigma]_q$}
\AxiomC{$l\simeq r$}
\BinaryInfC{$_pC[r\theta]_p[u\sigma]_q$}
\AxiomC{$u\simeq w$}
\BinaryInfC{$_qC[r\theta]_p[w\sigma]_q$}
\noLine
\UnaryInfC{$\vdots$}
\end{prooftree}
We have that $S(\Pi')=\langle p_1,...,p_n,p',p,q,q_1,...,q_m\rangle$. We have $S(\Pi')\lessdot S(\Pi)$, so $S(\Pi)$ is not the minimal, contradiction. Hence, by well-foundedness of $\lessdot$ we conclude that there is a $\ReC^\to$ derivation $\mathcal{C}\vdash^*D$, so $\ReC^\to$ admits ED.
\end{proof}

\subsection*{The $\ReC^\to_\lor$ calculus}

\begin{center}
\begin{tabular}{c p{.1\linewidth} l}
\multirow{3}{*}{
\AxiomC{$\drw{_p\underline{s[u]\bowtie t}\lor C}$}
\AxiomC{$\drw{_q\underline{l\simeq r}\lor D}$}
\LeftLabel{(\mtt{Sup})}
\BinaryInfC{$\drw{_\epsilon(s[r]\bowtie t\lor C \lor D)\theta}$}
\DisplayProof}
&
\multirow{3}{*}{where} & (1) $u$ is not a variable,\\
& & (2) $\theta=\mgu(l,u)$,\\
& & (3) $r\theta\not\succeq l\theta$ and $t\theta\not\succeq s\theta$,\\
\\

\multirow{2}{*}{
\AxiomC{$\drw{_p\underline{s\not\simeq t}\lor C}$}
\LeftLabel{(\mtt{EqRes})}
\UnaryInfC{$\drw{_\epsilon C\theta}$}
\DisplayProof}
&
\multirow{2}{*}{where} & \multirow{2}{*}{$\theta=\mgu(s,t)$,}\\
\\
\\

\multirow{2}{*}{
\AxiomC{$\drw{_p\underline{s\simeq t}\lor \underline{u\simeq w} \lor C}$}
\LeftLabel{(\mtt{EqFac})}
\UnaryInfC{$\drw{_\epsilon(s\simeq t\lor t\not\simeq w \lor C)\theta}$}
\DisplayProof}
&
\multirow{2}{*}{where} & (1) $\theta=\mgu(s,u)$,\\
& & (2) $t\theta\not\succeq s\theta$ and $w\theta\not\succeq t\theta$,\\
\\
\\

\multirow{2}{*}{\AxiomC{$\drw{_pC[l\theta]_q}$}
\AxiomC{$_{p'}l\simeq r$}
\LeftLabel{$(\Rw^\to_\downarrow)$}
\BinaryInfC{$\drw{_qC[r\theta]_q}$}
\DisplayProof} & \multirow{2}{*}{where} & (1) $l\theta\npreceq r\theta$,\\
& & (2) $q\nless_l p$,\\
\\
\\

\multirow{2}{*}{\AxiomC{$\erw{_pC[l\theta]_q}$}
\AxiomC{$_{p'}l\simeq r$}
\LeftLabel{$(\Rw^\to_\uparrow)$}
\BinaryInfC{$\urw{_qC[r\theta]_q}$}
\DisplayProof} & \multirow{2}{*}{where} & (1) $l\theta\nsucceq r\theta$,\\
& & (2) $\erw{_pC[l\theta]}=\drw{_pC[l\theta]}$ or $q\nless_l p$.\\
\medskip
\end{tabular}
\end{center}

\recToLorEgd*
\begin{proof}
The proof has similar ideas as the proof of Theorem~\ref{thm:rec_to_egd}. Since each $\ReC_\lor$ derivation is a sequence of $\Rw_\downarrow$ inferences followed by a sequence of $\Rw_\uparrow$ inferences, we can argue about the two sequences separately. First, we show that for each $\ReC_\lor$ derivation, there is a derivation where each $\Rw_\downarrow$ inference is an $\Rw^\to_\downarrow$  inference but not each $\Rw_\uparrow$ inference is necessarily an $\Rw^\to_\uparrow$ inference. First, we associate a sequence of positions $S(\Pi)$ to each derivation $\Pi$ only based on $\Rw_\downarrow$ inferences. We use the same $\lessdot$ ordering as in the proof of Theorem~\ref{thm:rec_to_egd}. Let $\mathcal{C}$ be a set of equations and $D$ a clause. Take a $\ReC_\lor$ derivation $\Pi=\mathcal{C}\vdash^* D$ s.t. $S(\Pi)$ is minimal w.r.t. $\lessdot$ but there is some $\Rw_\downarrow$ inference which is not an $\Rw^\to_\downarrow$ inference. This means that there are two consecutive $\Rw_\downarrow$ inferences of the following form:
\begin{center}
\AxiomC{$\vdots$}
\noLine
\UnaryInfC{$\drw{_{p'}C[l\theta]_p[u\sigma]_q}$}
\AxiomC{$u\simeq w$}
\BinaryInfC{$\drw{_qC[l\theta]_p[w\sigma]_q}$}
\AxiomC{$l\simeq r$}
\BinaryInfC{$\drw{_pC[r\theta]_p[w\sigma]_q}$}
\noLine
\UnaryInfC{$\vdots$}
\DisplayProof
\end{center}
where $p<_l q$ and we have $S(\Pi)=\langle p_1,...,p_n,p',q,p,q_1,...,q_m\rangle$. We switch the order of inferences as follows:
\begin{center}
\AxiomC{$\vdots$}
\noLine
\UnaryInfC{$\drw{_{p'}C[l\theta]_p[u\sigma]_q}$}
\AxiomC{$l\simeq r$}
\BinaryInfC{$\drw{_pC[r\theta]_p[u\sigma]_q}$}
\AxiomC{$u\simeq w$}
\BinaryInfC{$\drw{_qC[r\theta]_p[w\sigma]_q}$}
\noLine
\UnaryInfC{$\vdots$}
\DisplayProof
\end{center}
We get $\langle p_1,...,p_n,p',p,q,q_1,...,q_m\rangle<\langle p_1,...,p_n,p',q,p,q_1,...,q_m\rangle$, hence $S(\Pi)$ was not minimal w.r.t. $\lessdot$, contradiction. Hence, by well-foundedness of $\lessdot$, we conclude that there exists a $\ReC_\lor$ derivation $\mathcal{C}\vdash^*D$ where each $\Rw_\downarrow$ inference is an $\Rw^\to_\downarrow$ inference.

Next, we show that given a $\ReC_\lor$ derivation $\Pi=\mathcal{C}\vdash^*D$ where each $\Rw_\downarrow$ inference is an $\Rw^\to_\downarrow$ inference, there is a $\ReC^\to_\lor$ derivation. Now, we define $S(\Pi)$ only based on $\Rw_\uparrow$ inferences. Again, we assume there is a $\ReC_\lor$ derivation $\Pi=\mathcal{C}\vdash^* D$ s.t. $S(\Pi)$ is minimal w.r.t. $\lessdot$ but there is some $\Rw_\uparrow$ inference which is not an $\Rw^\to_\uparrow$ inference and there are two consecutive $\Rw_\uparrow$ inferences of the following form:
\begin{center}
\AxiomC{$\vdots$}
\noLine
\UnaryInfC{$\erw{_{p'}C[l\theta]_p[u\sigma]_q}$}
\AxiomC{$u\simeq w$}
\BinaryInfC{$\urw{_qC[l\theta]_p[w\sigma]_q}$}
\AxiomC{$l\simeq r$}
\BinaryInfC{$\urw{_pC[r\theta]_p[w\sigma]_q}$}
\noLine
\UnaryInfC{$\vdots$}
\DisplayProof
\end{center}
where $p<_l q$ and we have $S(\Pi)=\langle p_1,...,p_n,p',q,p,q_1,...,q_m\rangle$. We switch the order of inferences:
\begin{center}
\AxiomC{$\vdots$}
\noLine
\UnaryInfC{$\erw{_{p'}C[l\theta]_p[u\sigma]_q}$}
\AxiomC{$l\simeq r$}
\BinaryInfC{$\urw{_pC[r\theta]_p[u\sigma]_q}$}
\AxiomC{$u\simeq w$}
\BinaryInfC{$\urw{_qC[r\theta]_p[w\sigma]_q}$}
\noLine
\UnaryInfC{$\vdots$}
\DisplayProof
\end{center}
We get $\langle p_1,...,p_n,p',p,q,q_1,...,q_m\rangle<\langle p_1,...,p_n,p',q,p,q_1,...,q_m\rangle$, hence $S(\Pi)$ was not minimal w.r.t. $\lessdot$, contradiction. Hence, by well-foundedness of $\lessdot$, we conclude that there exists a $\ReC^\to_\lor$ derivation $\mathcal{C}\vdash^*D$. Hence, $\ReC^\to_\lor$ admits ED.
\end{proof}

\newpage
\section*{Appendix C}

\subsection*{Redundancy elimination}

\redundantIndInferenceI*
\begin{proof}
Since $\mathcal{I}$ admits ED, and by assumption there is a derivation $\mathcal{C}\vdash^*_\mathcal{I} (L[t])[l\sigma]\lor C$ for some set of clauses $\mathcal{C}$. Then, there is also a derivation of the form $\mathcal{C},l\simeq r\vdash^*_\mathcal{I} (L[t])[l\sigma\mapsto r\sigma]\lor C$. From $l\theta\succ r\theta$, we have $(L[t])[l\sigma]\lor C\succ (L[t])[l\sigma\mapsto r\sigma]\lor C$. There is also an $\Ind_G$ inference
$$(L[t])[l\sigma\mapsto r\sigma]\lor C\vdash \cnf(\neg G[(L[y])[l\sigma\mapsto r\sigma]]\lor C)$$
where $y$ is fresh. The two formulas $\neg G[L[x]]\lor C$ and $\neg G[(L[y])[l\sigma\mapsto r\sigma]]\lor C$ are equivalent w.r.t. the equation $l\simeq r$. Hence, the inference is redundant.
\end{proof}

\redundantIndInferenceII*
\begin{proof}
Since $\mathcal{I}$ admits ED and there is a derivation $\mathcal{C}\vdash^*_\mathcal{I}\overline{L}[t]\lor C$, then there is a derivation $\mathcal{C},l\simeq r\vdash^*_\mathcal{I}\overline{L}[t[l\theta\mapsto r\theta]]\lor C$. By $l\theta\succ r\theta$, we have $\overline{L}[t]\lor C\succ \overline{L}[t[l\theta\mapsto r\theta]]\lor C$. Since $\overline{L}[t\mapsto x]$ is equivalent to $\overline{L}[t[l\theta\mapsto r\theta]\mapsto y]$ (where $y$ is fresh) up to variable renaming, we have following the $\Ind_G$ inference
$$\overline{L}[t[l\theta\mapsto r\theta]]\lor C\vdash \cnf(\neg G[L[y]]\lor C).$$
The formula $\neg G[L[x]]\lor C)$ is equivalent to $\neg G[L[y]]\lor C$. Hence, we conclude that the $\Ind_G$ inference is redundant w.r.t. $\overline{L}[t[l\theta\mapsto r\theta]]\lor C$.
\end{proof}

\redundantWeaklyUsefulRewrite*
\begin{proof}
By the condition $C[l\theta]\succ (l\simeq r)\theta$ it is straightforward to show that the clause $C[l\theta]$ is redundant w.r.t. $C[r\theta]$ and $l\simeq r$, and that by the ED property $C[r\theta]$ can be derived. Take an arbitrary induction on $C[l\theta]$, inducting on term $t$. We consider the following cases:
\begin{enumerate}[label=(\roman*),leftmargin=1.7em]
\item If $t\triangleleft l\theta$, $l\simeq r$ being ineffective, all occurrences of $t$ inside $l\theta$ must be under a single variable position in $l\theta$ (since $l$ has only one occurrence of each variable). Let $\sigma=\{y\mapsto r[t\mapsto x]\mid y\mapsto r\in\theta\}$ for some fresh variable $x$. Since $l\succ r$, we also have $l\sigma\succ r\sigma$ and Lemma~\ref{reducibility-lemma} applies with $l\sigma$, hence the $\Ind_G$ inference is redundant.
\item If $l\theta\trianglelefteq t$, by Lemma~\ref{above-position-lemma} the inference is redundant.
\item Otherwise, $l\theta$ is in a parallel position to $t$ and Lemma~\ref{reducibility-lemma} applies with $l\theta$.
\end{enumerate}
\end{proof}



\subsection*{The \CReC calculus}

\begin{center}
\begin{tabular}{c p{.1\linewidth} l}
\multirow{3}{*}{
\AxiomC{$\underline{s[u]\bowtie t}\lor C$}
\AxiomC{$\underline{l\simeq r}\lor D$}
\LeftLabel{(\mtt{Sup})}
\BinaryInfC{$(s[r]\bowtie t\lor C \lor D)\theta$}
\DisplayProof}
&
\multirow{3}{*}{where} & (1) $u$ is not a variable,\\
& & (2) $\theta=\mgu(l,u)$,\\
& & (3) $r\theta\not\succeq l\theta$ and $t\theta\not\succeq s\theta$,\\[1em]

\multirow{2}{*}{
\AxiomC{$\underline{s\not\simeq t}\lor C$}
\LeftLabel{(\mtt{EqRes})}
\UnaryInfC{$C\theta$}
\DisplayProof}
&
\multirow{2}{*}{where} & \multirow{2}{*}{$\theta=\mgu(s,t)$,}\\[2.4em]

\multirow{2}{*}{
\AxiomC{$\underline{s\simeq t}\lor \underline{u\simeq w} \lor C$}
\LeftLabel{(\mtt{EqFac})}
\UnaryInfC{$(s\simeq t\lor t\not\simeq w \lor C)\theta$}
\DisplayProof}
&
\multirow{2}{*}{where} & (1) $\theta=\mgu(s,u)$,\\
& & (2) $t\theta\not\succeq s\theta$ and $w\theta\not\succeq t\theta$,\\[1.5em]

\multirow{2}{*}{
\AxiomC{$C[l\theta]$}
\AxiomC{$l\simeq r$}
\LeftLabel{$(\mtt{CRw})$}
\BinaryInfC{$C[r\theta]$}
\DisplayProof}
& \multirow{2}{*}{where} & \multirow{2}{*}{$l\theta\nprec r\theta$ or $l\simeq r$ is effective,}\\[2.4em]

\multirow{3}{*}{
\AxiomC{$s[l']\simeq t$}
\AxiomC{$l\simeq r$}
\LeftLabel{$(\Cleft)$}
\BinaryInfC{$(s[r]\simeq t)\theta$}
\DisplayProof}
& \multirow{3}{*}{where} & (1) $\theta=\mgu(l,l')$,\\
& & (2) $s[l']\simeq t$ is ineffective,\\
& & (3) $l\theta\not\succeq r\theta$ and $l\simeq r$ is effective,\\[1em]

\multirow{3}{*}{
\AxiomC{$s\simeq t[l']$}
\AxiomC{$l\simeq r$}
\LeftLabel{$(\Cright)$}
\BinaryInfC{$(s\simeq t[r])\theta$}
\DisplayProof}
& \multirow{3}{*}{where} & (1) $\theta=\mgu(l,l')$,\\
& & (2) $l\simeq r$ is ineffective,\\
& & (3) $t[l']\theta\not\succeq s\theta$ and $s\simeq t[l']$ is effective.\\[1em]

\end{tabular}
\end{center}

\begin{lemma}
\label{lemma:linear-rw-derivation}
Let $\mathcal{C}$ be a set of clauses and $D$ a clause. If there is a derivation $\mathcal{C}\vdash^*_{\{\Rw\}} D$, then there is a derivation of the following form:
\begin{prooftree}
\AxiomC{$C_1$}
\AxiomC{$l_1\simeq r_1$}
\LeftLabel{$(\Rw)$}
\BinaryInfC{$C_2$}
\AxiomC{$l_2\simeq r_2$}
\LeftLabel{$(\Rw)$}
\BinaryInfC{$C_3$}
\AxiomC{}
\noLine
\BinaryInfC{$\ddots$}
\noLine
\UnaryInfC{$C_n$}
\AxiomC{$l_n\simeq r_n$}
\LeftLabel{$(\Rw)$}
\BinaryInfC{$D$}
\end{prooftree}
where $l_1\simeq r_1,...,l_n\simeq r_n\in\mathcal{C}$.
\end{lemma}
\begin{proof}
If a derivation is not of the desired form, we have one of the following cases in the derivation:
\begin{center}
(i)
\AxiomC{$C[l[w\sigma]\theta]$}
\AxiomC{$l[u\sigma]\simeq r$}
\AxiomC{$u\simeq w$}
\LeftLabel{$(\Rw)$}
\BinaryInfC{$l[w\sigma]\simeq r$}
\LeftLabel{$(\Rw)$}
\BinaryInfC{$C[r\theta]$}
\DisplayProof
\qquad
(ii)
\AxiomC{$C[l[w\sigma]\theta]$}
\AxiomC{$l\simeq r[u\sigma]$}
\AxiomC{$u\simeq w$}
\LeftLabel{$(\Rw)$}
\BinaryInfC{$l\simeq r[w\sigma]$}
\LeftLabel{$(\Rw)$}
\BinaryInfC{$C[r[w\sigma]\theta]$}
\DisplayProof
\end{center}
We transform the two cases into the following derivations:
\begin{center}
(i')
\AxiomC{$C[l[w\sigma]\theta]$}
\AxiomC{$u\simeq w$}
\LeftLabel{$(\Rw)$}
\BinaryInfC{$C[l[u\sigma]\theta]$}
\AxiomC{$l[u\sigma]\simeq r$}
\LeftLabel{$(\Rw)$}
\BinaryInfC{$C[r\theta]$}
\DisplayProof
\qquad
(ii')
\AxiomC{$C[l\theta]$}
\AxiomC{$l\simeq r[u\sigma]$}
\LeftLabel{$(\Rw)$}
\BinaryInfC{$C[r[u\sigma]\theta]$}
\AxiomC{$u\simeq w$}
\LeftLabel{$(\Rw)$}
\BinaryInfC{$C[r[w\sigma]\theta]$}
\DisplayProof
\end{center}
It is easy to see that repeating this transformation results in a derivation of the desired form.
\end{proof}

\recChaining*
\begin{proof}



In this proof, we call \emph{violating inferences} any $\Rw$ inferences that are not $\mtt{CRw}$ inferences. The proof is by induction on the number of violating inferences. Take an arbitrary derivation $\mathcal{C}\vdash^*_{\{\Rw\}} D$ s.t. $D$ is an inductively non-redundant clause. Let us denote this derivation with $\Pi$. If $\Pi$ does not contain any violating inferences, we are done as $\Pi$ is also a $\{\mtt{CRw},\Cleft,\Cright\}$ derivation. 

Otherwise, by Lemma~\ref{lemma:linear-rw-derivation} we may assume that $\Pi$ is in a form where \Rw inferences are only applied on a single clause repeatedly. Take the last (in the order of \Rw inferences) violating inference in $\Pi$ rewriting term $r\theta$ into $l\theta$ with ineffective $l\simeq r$ and $l\theta\succ r\theta$. We build a new derivation $\Gamma$ by induction on the derivation length of $\Pi$ s.t. it does not contain the violating inference. We claim that for each inference in $\Pi$ resulting in some clause $C[l\theta_1]...[l\theta_n]$ for some substitutions $\theta_1$,...,$\theta_n$, there is a derivation resulting in $C[r\theta_1]...[r\theta_n]$. Let us assume that we have a derivation $\Gamma$ corresponding to $n-1$ inferences in $\Pi$. We take the $n$th inference in $\Pi$ and consider the following cases.
\begin{enumerate}[label=(\roman*),leftmargin=1.7em]
\item If the clause $C$ we have in $\Pi$ is the same as the clause in $\Gamma$, we either have the violating inference, in which case we skip the inference and we get $C[r\theta]$ instead of $C[l\theta]$ or we have some other inference and we perform it, resulting in the same clause as in the conclusion of the $n$th inference of $\Pi$. 
\item Otherwise, in $\Pi$ the premise of the $n$th step is of the form $C[l\theta_1]...[l\theta_n]$ but the premise in $\Gamma$ is $C[r\theta_1]...[r\theta_n]$. If the inference is an \Rw inference parallel to all $l\theta_i$s in $\Pi$ resulting in some $C'[l\theta_1]...[l\theta_n]$, then we simply perform the inference in $\Gamma$ as well, resulting in the desired $C'[r\theta_1]...[r\theta_n]$.
\item If the inference is an \Rw inference on a subterm of some $l\theta_i$ in $\Pi$ and at or below a variable position in $l$, then let $x$ be this variable. Since $l$ has only one occurrence of $x$, the result is $l\theta'_i$ for some $\theta'_i$. In the new derivation, we perform the same rewrite inside each occurrence of $x$ in $r\theta_i$, resulting in $r\theta'_i$. Hence, in $\Pi$ we get $C[l\theta_1]...[l\theta'_i]...[l\theta_n]$, and in $\Gamma$ we get the desired $C[r\theta_1]...[r\theta'_i]...[r\theta_n]$.
\item If the inference is an \Rw inference with some equation $u\simeq w$ s.t. it rewrites some subterm $u\sigma]$ above a variable position of an $l\theta_i$ term in $\Pi$ for some substitution $\sigma$, then there is an inference
\begin{prooftree}
\AxiomC{$l[u']\simeq r$}
\AxiomC{$u\simeq w$}
\LeftLabel{$\kappa$}
\BinaryInfC{$(l[w]\simeq r)\rho$}
\end{prooftree}
where $\rho=\mgu(u',u)$. We have $l[u']\rho\succ r\rho$ by assumption. If $u\rho=w\rho$, then the rewrite with $u\simeq w$ in $l\theta_i$ results in the same clause $C[l\theta_1]...[l\theta_n]$ in $\Pi$, so by not performing the inference in $\Gamma$ we get the desired clause $C[r\theta_1]...[r\theta_n]$ trivially. Otherwise, if $w\rho\nsucceq u\rho$, then the inference $\kappa$ is a $\mtt{Sup}$ inference.  If $u\rho\nsucceq w\rho$, then the inference $\kappa$ is a \Cleft inference (since $u\simeq w$ cannot be an ineffective equation, by assumption). Hence, we can use $(l[w]\simeq r)\rho$ to rewrite $r\theta_i$ instead, resulting in $C[r\theta_1]...[r\theta_{i-1}][r\theta_{i+1}]...[l\theta_n]$ as desired.
\item Otherwise, the inference is an \Rw inference s.t. the rewritten term is the superterm of some $l\theta_i$s. Let $u\simeq w$ be the rewriting equation and $u\sigma$ be the rewritten term in $\Pi$. Let $p$ be the position of $u\sigma$ in $C[l\theta_1]...[l\theta_n]$. Let $u'$ be the term in the same position $p$ in $C[r\theta_1]...[r\theta_n]$. We induct on the number of positions inside $u'$ which prevent it from being rewritten by $u\simeq w$ (i.e. the positions preventing $u'$ being matched by $u$). Note that these terms are all $l\theta_i$s. If we have an $l\theta_i$ in a position in $u\sigma$ which is above a variable position in $u$, then there is an inference
\begin{prooftree}
\AxiomC{$u[l']\simeq w$}
\AxiomC{$l\simeq r$}
\LeftLabel{$\kappa$}
\BinaryInfC{$(u[r]\simeq w)\rho$}
\end{prooftree}
where $\rho=\mgu(l,l')$. Again, we have $l\rho\succ r\rho$ by assumption. Similarly as in the previous case, if $u[l']\rho=w\rho$, then the rewrite with $u\simeq w$ on $u\sigma$ results in the same clause $C[l\theta_1]...[l\theta_n]$ in $\Pi$, so by not performing the inference in $\Gamma$ we get the desired clause $C[r\theta_1]...[r\theta_n]$ trivially. If $w\rho\nsucceq u[l']\rho$, then $\kappa$ is a $\mtt{Sup}$ inference. If $u[l']\rho\nsucceq w\rho$, then $\kappa$ is a \Cright inference. By using $(u[r]\simeq w)\rho$ instead of $u[l']\simeq w$ for the rewrite, we get one less position where the rewrite (or match) is prevented, hence the induction hypothesis applies.

If there is no such $l\theta_i$, then either we have $u\sigma'$ for some $\sigma'$ and we can apply the rewrite resulting in $w\sigma'$, and we get a clause of the desired form, or $u$ does not match $u'$ because there are two occurrences of a variable $x$ in $u$ s.t. there are two distinct terms $t_1$ and $t_2$ in these positions in $u'$. This must be because $t_1$ contains $r\theta_i$ in some position where $t_2$ contains $l\theta_i$. We rewrite all such $l\theta_i$s in $u'$ into $r\theta_i$ using downward rewrites with $l\simeq r$, resulting in a term that can be rewritten by $u\simeq w$.
\end{enumerate}
We have covered all cases, proving the claim that all $n$ inferences in $\Pi$ can be performed or replaced by suitable inferences in $\Gamma$ resulting in some $D[r\theta_1]...,[r\theta_n]$ instead of $D=D[l\theta_1]...[l\theta_n]$. But $D$ is inductively non-redundant, so it cannot contain any such $l\theta_i$s. Hence, we get $D$ as the result of $\Gamma$ as well, with $\Gamma$ containing one less violating inference. We apply the induction hypothesis, and get a \CReC derivation.
\end{proof}

\newpage
\subsection*{The $\CReC_\lor$ calculus}

\begin{center}
\begin{tabular}{c p{.1\linewidth} l}
\multirow{3}{*}{
\AxiomC{$\drw{\underline{s[u]\bowtie t}\lor C}$}
\AxiomC{$\drw{\underline{l\simeq r}\lor D}$}
\LeftLabel{(\mtt{Sup})}
\BinaryInfC{$\drw{(s[r]\bowtie t\lor C \lor D)\theta}$}
\DisplayProof}
&
\multirow{3}{*}{where} & (1) $u$ is not a variable,\\
& & (2) $\theta=\mgu(l,u)$,\\
& & (3) $r\theta\not\succeq l\theta$ and $t\theta\not\succeq s\theta$,\\
\\

\multirow{2}{*}{
\centering
\AxiomC{$\drw{\underline{s\not\simeq t}\lor C}$}
\LeftLabel{(\mtt{EqRes})}
\UnaryInfC{$\drw{C\theta}$}
\DisplayProof}
&
\multirow{2}{*}{where} & \multirow{2}{*}{$\theta=\mgu(s,t)$,}\\
\\
\\

\multirow{2}{.45\linewidth}{
\centering
\AxiomC{$\drw{\underline{s\simeq t}\lor \underline{u\simeq w} \lor C}$}
\LeftLabel{(\mtt{EqFac})}
\UnaryInfC{$\drw{(s\simeq t\lor t\not\simeq w \lor C)\theta}$}
\DisplayProof}
&
\multirow{2}{*}{where} & (1) $\theta=\mgu(s,u)$,\\
& & (2) $t\theta\not\succeq s\theta$ and $w\theta\not\succeq t\theta$,\\[1.8em]

\multirow{2}{*}{\AxiomC{$\drw{C[l\theta]}$}
\AxiomC{$\drw{l\simeq r}$}
\LeftLabel{$(\mtt{CRw}_\downarrow)$}
\BinaryInfC{$\drw{C[r\theta]}$}
\DisplayProof} & \multirow{2}{*}{where} &\multirow{2}{*}{$l\theta\npreceq r\theta$,}\\[2.4em]

\multirow{2}{*}{\AxiomC{$\erw{C[l\theta]}$}
\AxiomC{$\drw{l\simeq r}$}
\LeftLabel{$(\mtt{CRw}_\uparrow)$}
\BinaryInfC{$\urw{C[r\theta]}$}
\DisplayProof} & \multirow{2}{*}{where} &(1) $l\theta\nsucceq r\theta$,\\
& & (2) $l\theta\nprec r\theta$ or $l\simeq r$ is effective.\\[1em]

\multirow{3}{*}{
\AxiomC{$\drw{s[l']\simeq t}$}
\AxiomC{$\drw{l\simeq r}$}
\LeftLabel{$(\Cleft)$}
\BinaryInfC{$\drw{(s[r]\simeq t)\theta}$}
\DisplayProof}
& \multirow{3}{*}{where} & (1) $\theta=\mgu(l,l')$,\\
& & (2) $s[l']\simeq t$ is ineffective,\\
& & (3) $l\theta\not\succeq r\theta$ and $l\simeq r$ is effective,\\[1em]

\multirow{3}{*}{
\AxiomC{$\drw{s\simeq t[l']}$}
\AxiomC{$\drw{l\simeq r}$}
\LeftLabel{$(\Cright)$}
\BinaryInfC{$\drw{(s\simeq t[r])\theta}$}
\DisplayProof}
& \multirow{3}{*}{where} & (1) $\theta=\mgu(l,l')$,\\
& & (2) $l\simeq r$ is ineffective,\\
& & (3) $t[l']\theta\not\succeq s\theta$ and $s\simeq t[l']$ is effective.\\[1em]

\end{tabular}
\medskip
\end{center}

\begin{theorem}[$\CReC_\lor$ derivability]
\label{thm:crec_derivability}
Given an inductively non-redundant clause $D$, if $\mathcal{C}\vdash^*_{\{\Rw_\downarrow,\Rw_\uparrow\}}D$, then $\mathcal{C}\vdash^*_{\{\CRw_\downarrow,\CRw_\uparrow,\Cleft,\Cright\}}D$.
\end{theorem}
\begin{proof}
Since each $\Rw_\downarrow$ and $\Rw_\uparrow$ inference is also an $\Rw$ inference, the derivation is also a $\{\Rw\}$ derivation and Theorem~\ref{thm:chaining} applies. Hence we get a derivation $\mathcal{C}\vdash^*_{\{\CRw,\Cleft,\Cright\}}D$. From this derivation, similarly to the proof of Theorem~\ref{thm:rw_v_ugd}, we obtain a $\{\CRw_\downarrow,\CRw_\uparrow,\Cleft,\Cright\}$ derivation.
\end{proof}

\newpage
\subsection*{The $\CReC^\to$ calculus}

\begin{center}
\begin{tabular}{c p{.1\linewidth} l}
\multirow{3}{*}{
\AxiomC{$_p\underline{s[u]\bowtie t}\lor C$}
\AxiomC{$_q\underline{l\simeq r}\lor D$}
\LeftLabel{(\mtt{Sup})}
\BinaryInfC{$_\epsilon(s[r]\bowtie t\lor C \lor D)\theta$}
\DisplayProof}
&
\multirow{3}{*}{where} & (1) $u$ is not a variable,\\
& & (2) $\theta=\mgu(l,u)$,\\
& & (3) $r\theta\not\succeq l\theta$ and $t\theta\not\succeq s\theta$,\\
\\

\multirow{2}{*}{
\centering
\AxiomC{$_p\underline{s\not\simeq t}\lor C$}
\LeftLabel{(\mtt{EqRes})}
\UnaryInfC{$_\epsilon C\theta$}
\DisplayProof}
&
\multirow{2}{*}{where} & \multirow{2}{*}{$\theta=\mgu(s,t)$,}\\
\\
\\

\multirow{2}{.45\linewidth}{
\centering
\AxiomC{$_p\underline{s\simeq t}\lor \underline{u\simeq w} \lor C$}
\LeftLabel{(\mtt{EqFac})}
\UnaryInfC{$_\epsilon(s\simeq t\lor t\not\simeq w \lor C)\theta$}
\DisplayProof}
&
\multirow{2}{*}{where} & (1) $\theta=\mgu(s,u)$,\\
& & (2) $t\theta\not\succeq s\theta$ and $w\theta\not\succeq t\theta$,\\[1.8em]

\multirow{2}{*}{
\AxiomC{$_pC[l\theta]_q$}
\AxiomC{$_{p'}l\simeq r$}
\LeftLabel{$(\mtt{CRw}^\to)$}
\BinaryInfC{$_qC[r\theta]_q$}
\DisplayProof} & \multirow{2}{*}{where} & (1) $q\nless_l p$,\\
& & (2) $l\theta\nprec r\theta$ or $l\simeq r$ is effective.\\[1em]

\multirow{3}{*}{
\AxiomC{$_p s[l']\simeq t$}
\AxiomC{$_q l\simeq r$}
\LeftLabel{$(\Cleft)$}
\BinaryInfC{$_\epsilon(s[r]\simeq t)\theta$}
\DisplayProof}
& \multirow{3}{*}{where} & (1) $\theta=\mgu(l,l')$,\\
& & (2) $s[l']\simeq t$ is ineffective,\\
& & (3) $l\theta\not\succeq r\theta$ and $l\simeq r$ is effective,\\[1em]

\multirow{3}{*}{
\AxiomC{$_p s\simeq t[l']$}
\AxiomC{$_q l\simeq r$}
\LeftLabel{$(\Cright)$}
\BinaryInfC{$_\epsilon(s\simeq t[r])\theta$}
\DisplayProof}
& \multirow{3}{*}{where} & (1) $\theta=\mgu(l,l')$,\\
& & (2) $l\simeq r$ is ineffective,\\
& & (3) $t[l']\theta\not\succeq s\theta$ and $s\simeq t[l']$ is effective.\\[1em]

\end{tabular}
\medskip
\end{center}

\begin{theorem}[$\CReC^\to$ derivability]
Given an inductively non-redundant clause $D$, if $\mathcal{C}\vdash^*_{\{\Rw^\to\}}D$, then $\mathcal{C}\vdash^*_{\{\CRw^\to,\Cleft,\Cright\}}D$.
\end{theorem}
\begin{proof}
Similarly to the proof of Theorem~\ref{thm:crec_derivability}, for any $\{\Rw\}$ derivation, we get a $\{\CRw,\Cleft,\Cright\}$ derivation by Theorem~\ref{thm:chaining}. We then use the same approach as in Theorem~\ref{thm:rec_to_egd} to get the desired derivation.
\end{proof}

\newpage
\subsection*{The $\CReC^\to_\lor$ calculus}

\begin{center}
\begin{tabular}{c p{.1\linewidth} l}
\multirow{3}{*}{
\AxiomC{$\drw{_p\underline{s[u]\bowtie t}\lor C}$}
\AxiomC{$\drw{_q\underline{l\simeq r}\lor D}$}
\LeftLabel{(\mtt{Sup})}
\BinaryInfC{$\drw{_\epsilon(s[r]\bowtie t\lor C \lor D)\theta}$}
\DisplayProof}
&
\multirow{3}{*}{where} & (1) $u$ is not a variable,\\
& & (2) $\theta=\mgu(l,u)$,\\
& & (3) $r\theta\not\succeq l\theta$ and $t\theta\not\succeq s\theta$,\\
\\

\multirow{2}{*}{
\centering
\AxiomC{$\drw{_p\underline{s\not\simeq t}\lor C}$}
\LeftLabel{(\mtt{EqRes})}
\UnaryInfC{$\drw{_\epsilon C\theta}$}
\DisplayProof}
&
\multirow{2}{*}{where} & \multirow{2}{*}{$\theta=\mgu(s,t)$,}\\
\\
\\

\multirow{2}{.45\linewidth}{
\centering
\AxiomC{$\drw{_p\underline{s\simeq t}\lor \underline{u\simeq w} \lor C}$}
\LeftLabel{(\mtt{EqFac})}
\UnaryInfC{$\drw{_\epsilon(s\simeq t\lor t\not\simeq w \lor C)\theta}$}
\DisplayProof}
&
\multirow{2}{*}{where} & (1) $\theta=\mgu(s,u)$,\\
& & (2) $t\theta\not\succeq s\theta$ and $w\theta\not\succeq t\theta$,\\[1.8em]

\multirow{2}{*}{
\AxiomC{$\drw{_p C[l\theta]_q}$}
\AxiomC{$\drw{_{p'}l\simeq r}$}
\LeftLabel{$(\mtt{CRw}_\downarrow)$}
\BinaryInfC{$\drw{_q C[r\theta]_q}$}
\DisplayProof} & \multirow{2}{*}{where} & (1) $l\theta\npreceq r\theta$,\\
& & (2) $q\nless_l p$,\\[2.4em]

\multirow{2}{*}{
\AxiomC{$\erw{_p C[l\theta]_q}$}
\AxiomC{$\drw{_{p'} l\simeq r}$}
\LeftLabel{$(\mtt{CRw}_\uparrow)$}
\BinaryInfC{$\urw{_q C[r\theta]_q}$}
\DisplayProof} & \multirow{2}{*}{where} &(1) $l\theta\nsucceq r\theta$,\\
& & (2) $\erw{_pC[l\theta]}=\drw{_pC[l\theta]}$ or $q\nless_l p$,\\
& & (3) $l\theta\nprec r\theta$ or $l\simeq r$ is effective,\\[1em]

\multirow{3}{*}{
\AxiomC{$\drw{_p s[l']\simeq t}$}
\AxiomC{$\drw{_q l\simeq r}$}
\LeftLabel{$(\Cleft)$}
\BinaryInfC{$\drw{_\epsilon (s[r]\simeq t)\theta}$}
\DisplayProof}
& \multirow{3}{*}{where} & (1) $\theta=\mgu(l,l')$,\\
& & (2) $s[l']\simeq t$ is ineffective,\\
& & (3) $l\theta\not\succeq r\theta$ and $l\simeq r$ is effective,\\[1em]

\multirow{3}{*}{
\AxiomC{$\drw{_p s\simeq t[l']}$}
\AxiomC{$\drw{_q l\simeq r}$}
\LeftLabel{$(\Cright)$}
\BinaryInfC{$\drw{_\epsilon(s\simeq t[r])\theta}$}
\DisplayProof}
& \multirow{3}{*}{where} & (1) $\theta=\mgu(l,l')$,\\
& & (2) $l\simeq r$ is ineffective,\\
& & (3) $t[l']\theta\not\succeq s\theta$ and $s\simeq t[l']$ is effective.\\[1em]

\end{tabular}
\medskip
\end{center}

\begin{theorem}[$\CReC^\to_\lor$ derivability]
Given an inductively non-redundant clause $D$, if $\mathcal{C}\vdash^*_{\{\Rw^\to_\downarrow,\Rw^\to_\uparrow\}}D$, then $\mathcal{C}\vdash^*_{\{\CRw^\to_\downarrow,\CRw^\to_\uparrow,\Cleft,\Cright\}}D$.
\end{theorem}
\begin{proof}
Any $\{\Rw^\to_\downarrow,\Rw^\to_\uparrow\}$ derivation is a $\{\Rw_\downarrow,\Rw_\uparrow\}$ derivation, hence we get a $\{\CRw_\downarrow,\CRw_\uparrow,\Cleft,\Cright\}$ derivation by Theorem~\ref{thm:crec_derivability}. We then use the same approach as in Theorem~\ref{thm:rec_to_lor_egd} to get the desired derivation.
\end{proof}

\end{document}